\documentclass[sigconf]{acmart}

\AtBeginDocument{%
  }

\copyrightyear{2025}
\acmYear{2025}
\setcopyright{cc}
\setcctype{by}
\acmConference[MobiHoc '25]{The Twenty-sixth International Symposium on Theory, Algorithmic Foundations, and Protocol Design for Mobile Networks and Mobile Computing}{October 27--30, 2025}{Houston, TX, USA}
\acmBooktitle{The Twenty-sixth International Symposium on Theory, Algorithmic Foundations, and Protocol Design for Mobile Networks and Mobile Computing (MobiHoc '25), October 27--30, 2025, Houston, TX, USA}
\acmDOI{10.1145/3704413.3764415}
\acmISBN{979-8-4007-1353-8/2025/10}

\usepackage[utf8]{inputenc} 
\usepackage[T1]{fontenc}    
\usepackage{url}            
\usepackage{booktabs}       
\usepackage{amsfonts}       
\usepackage{nicefrac}       
\usepackage{microtype}      
\usepackage{xcolor}
\usepackage{verbatim}
\usepackage{graphicx}
\usepackage{textcomp}
\usepackage{cases}
\usepackage{amsmath,amsfonts}
\usepackage{fancyhdr,graphicx}
\usepackage{subfigure}
\usepackage[utf8]{inputenc}
\usepackage{soul}
\usepackage[noend]{algorithm, algorithmic}
\usepackage{multirow}
\usepackage{fontenc}
\usepackage{amsthm}
\usepackage{enumitem}
\newtheorem{theorem}{Theorem}
\newtheorem{lemma}{Lemma}

\newtheorem{definition}{Definition}

\usepackage{subcaption}
\newtheoremstyle{example_style}  
  {3pt}  
  {3pt}    
  {\normalfont}  
  {}       
  {\bfseries} 
  {.}      
  { }      
  {\thmname{#1}\thmnumber{ #2}\thmnote{ (#3)}}  

\theoremstyle{example_style}
\newtheorem{Example}{Example} 

\begin{document}

\title{Understanding the Fundamental Trade-Off Between Age of Information and Throughput in Unreliable Wireless Networks}
\author{Lin Wang}
\email{linwang0420@tamu.edu}
\affiliation{%
  \institution{Department of ECE, Texas A\& M University}
  \city{College Station}
  \state{Texas}
  \country{USA}
}
\author{I-Hong Hou}
\email{ihou@tamu.edu}
\affiliation{%
  \institution{Department of ECE, Texas A\& M University}
  \city{College Station}
  \state{Texas}
  \country{USA}
}

\begin{abstract}
This paper characterizes the fundamental trade-off between throughput and Age of Information (AoI) in wireless networks where multiple devices transmit status updates to a central base station over unreliable channels. To address the complexity introduced by stochastic transmission successes, we propose the throughput-AoI capacity region, which defines all feasible throughput-AoI pairs achievable under any scheduling policy. Using a second-order approximation that incorporates both mean and temporal variance, we derive an outer bound and a tight inner bound for the throughput-AoI capacity region. Furthermore, we propose a simple and low complexity scheduling policy and prove that it achieves every interior point within the tight inner bound. This establishes a systematic and theoretically grounded framework for the joint optimization of throughput and information freshness in practical wireless communication scenarios.

To validate our theoretical framework and demonstrate the utility of the throughput-AoI capacity region, extensive simulations are implemented. Simulation results demonstrate that our proposed policy significantly outperforms conventional methods across various practical network optimization scenarios. The findings highlight our approach's effectiveness in optimizing both throughput and AoI, underscoring its applicability and robustness in practical wireless networks.

\end{abstract}

\ccsdesc[500]{Networks~Packet scheduling}
\ccsdesc[500]{Networks~Network performance modeling}
\ccsdesc[500]{Networks~Network performance analysis}
\keywords{Wireless networks, Age of Information, Throughput, Scheduling}
\maketitle

\section{Introduction} \label{section:intro}

The increasing demand for real-time information updates in wireless networks has led to the emergence of Age of Information (AoI) as a critical performance metric\cite{firstaoi, aoisurvey}. Unlike conventional latency measures, AoI captures data freshness—a critical aspect for Internet of Things (IoT)\cite{iotaoi,underwateraoi,energysys1,energysys2}, cyber-physical systems \cite{cpsaoi,cpsaoi2}, and autonomous control networks \cite{uavaoi,uavaoi2,uavaoi3}—since outdated information can significantly impair functionality, particularly in time-sensitive applications such as emergency monitoring systems. Indeed, existing work\cite{ramakanth2024monitoring,ornee2019sampling} has shown that the estimation error of some random processes depends only on the AoI of the flow. However, in many other scenarios, the estimation accuracy of a remote sensing application depends on both the quantity and the freshness of available data\cite{guo2021scheduling}. In such scenarios, optimizing the estimation accuracy requires the joint optimization of throughput and AoI, which remains an open problem. The objective of this paper is to characterize the trade-off between throughput and AoI among heterogeneous devices in unreliable wireless networks precisely, which can then be used to address the joint optimization of throughput and AoI. 

To bridge the gap in understanding the relationship between AoI and throughput, we first propose an analytical system model that captures both throughput and AoI for each device. We formally define the throughput-AoI capacity region, which represents the set of all feasible throughput-AoI pairs achievable by any scheduling policy under given network constraints. Furthermore, we demonstrate that the throughput-AoI capacity region can be employed to address several practical network optimization problems, many of which remain unsolved.

To characterize the throughput-AoI capacity region, we adopt a second-order analysis approach that represents each random process by its mean and temporal variance. This approach provides a more accurate approximation of AoI dynamics than conventional first-order techniques. Moreover, leveraging the Markov Central Limit Theorem reduces the complexity associated with stochastic fluctuations in the network. To further streamline our analysis, we introduce the concept of system-wide temporal variance to simplify the evaluation of variance across individual devices.

Using the second-order analysis, we establish an outer bound and an inner bound for the throughput-AoI capacity region, with the two bounds being nearly identical except at the boundaries. Furthermore, we propose a Variance-Weighted Deficit (VWD) scheduling policy that is capable of achieving every point within the inner bound. The simplicity and low complexity of the VWD policy make it highly practical for implementation, providing an effective and efficient framework to balance throughput and information freshness in wireless networks.

To demonstrate the practical applicability of our theoretical results, we apply our framework to four representative problems: AoI minimization with hard throughput constraints, cost minimization with soft throughput constraints, proportional fairness in both throughput and AoI, and admission control for devices with AoI constraints. Extensive simulations confirm that our proposed scheduling policy significantly outperforms existing approaches, thereby addressing the critical gap in understanding the interplay between AoI and throughput in unreliable wireless networks.

Our contributions can be summarized as follows:
\begin{itemize}
\item We formalize the fundamental trade-off between throughput and AoI in unreliable wireless networks and define the corresponding throughput-AoI capacity region.
\item We derive nearly identical outer and inner bounds for the throughput-AoI capacity region and propose a low-complexity scheduling policy capable of achieving every interior point within this region.
\item We conduct extensive simulations across multiple optimization scenarios to validate the effectiveness of our proposed throughput-AoI capacity framework, demonstrating that our policy significantly outperforms existing optimal scheduling policies in balancing AoI and throughput.
\end{itemize}

The remainder of this paper is organized as follows: Section~\ref{section:relatedwork} reviews recent studies on wireless scheduling. Section~\ref{section:model} formally defines the system model and problem formulation. Section~\ref{section:secorderappro} introduces the second-order AoI approximation and discusses the system-wide variance. Section~\ref{section:outerbound} derives an outer bound for the throughput-AoI capacity region, establishing performance limits for all feasible scheduling policies. Section~\ref{section:innerbound} further defines a tight inner bound for the capacity region, proposes a simple scheduling policy, and analytically demonstrates its ability to achieve every interior point within the throughput-AoI capacity region. Section~\ref{section:simulation} presents comprehensive simulation results, verifying our theoretical findings and highlighting the advantages of our policy compared to existing methods. Section~\ref{section:conclusion} concludes the paper by summarizing key insights and outlining potential directions for future research.

\section{RELATED WORK}\label{section:relatedwork}
A significant amount of research has focused on designing scheduling policies in wireless networks to minimize the Age of Information~\cite{freshcsma,csma2,whittleminaoi,aoithreshold}. Kadota et al.\cite{broadcast}proposed a scheduling policy to minimize Age of Information in wireless broadcast networks with unreliable channels, proving the Greedy Policy's optimality in symmetric cases and deriving a closed-form Whittle Index for general settings. Kadota and Modiano~\cite{stocasticarrival} studied a system where a base station serves multiple traffic streams destined for different receivers, with packet arrivals following a stochastic process and being enqueued separately. They derived a lower bound on the achievable AoI under any queueing discipline and further developed an Optimal Stationary Randomized policy and a Max-Weight policy for three common queueing disciplines. Tripathi and Modiano~\cite{tripathi2022optimizing} developed a simple model for the timely monitoring of correlated sources over a wireless network. Furthermore, Pan et al.~\cite{twochannelminaoi} studied the problem of minimizing AoI when a source transmits status updates over two heterogeneous communication channels. Partial-Index approaches are also used in AoI minimization problems~\cite{partialindex,partialindex2} for multichannel scenarios. These studies share a common objective: to optimize the allocation of channel resources among networked sensors for enhanced AoI performance. Li et al.~\cite{aoiviolation}examined an AoI scheduling problem under a violation tolerance constraint that permits occasional AoI violation and defined the corresponding capacity region in terms of AoI deadlines, tolerance rates, and packet loss rates. Zhou and Lin\cite{whittleminaoi} established the Whittle indexability of the AoI minimization problem with stochastic arrivals and unreliable channels by introducing a novel, easily verifiable Active Time condition and a coupling-based proof approach. However, these works neglect throughput—a critical metric in communication networks—and thus overlook the trade-off between minimizing AoI and maintaining high throughput.

A key insight is that high throughput does not necessarily result in low AoI\cite{aoitprela}. To address this challenge, Kadota and Modiano~\cite{modiano}studied a single-hop wireless network where multiple nodes transmit time-sensitive information to a base station. They proposed three low-complexity scheduling policies—randomized, Max-Weight, and Whittle’s Index—to minimize the expected weighted sum AoI while ensuring timely throughput constraints. Fountoulakis et al.~\cite{aoiandtp} investigated a mixed-traffic wireless system consisting of both deadline-constrained and AoI-sensitive users and proposed an efficient scheduling policy that minimizes the average Age of Information while simultaneously satisfying timely throughput constraints. However, their study assumes that all devices have strict throughput requirements, which may not always be the case in practical scenarios. To the best of our knowledge, no existing research has systematically quantified the trade-off between throughput and AoI. 

Guo et al. ~\cite{daojingsecondorder, guo2024aoi} proposed a novel framework for optimizing the second-order behavior of wireless networks by capturing both the mean and temporal variance of random processes, unlike traditional approaches that focus solely on first-order statistics. Similarly, Fan et al.~\cite{siqi} characterized both active and passive users in a random access network by using its mean and temporal variance to minimize AoI. However, their applicability is limited to ON-OFF channel models or random access models for AoI minimization scenarios, restricting its generalizability to broader network settings.

\section{System Model}\label{section:model}
We consider a single-hop wireless network scenario, such as a smart home system, as shown in Fig.~\ref{fig:smarthome}.

\begin{figure}[h]
    \vspace{-10pt}
    \begin{center}
        \includegraphics[width=0.95\linewidth]{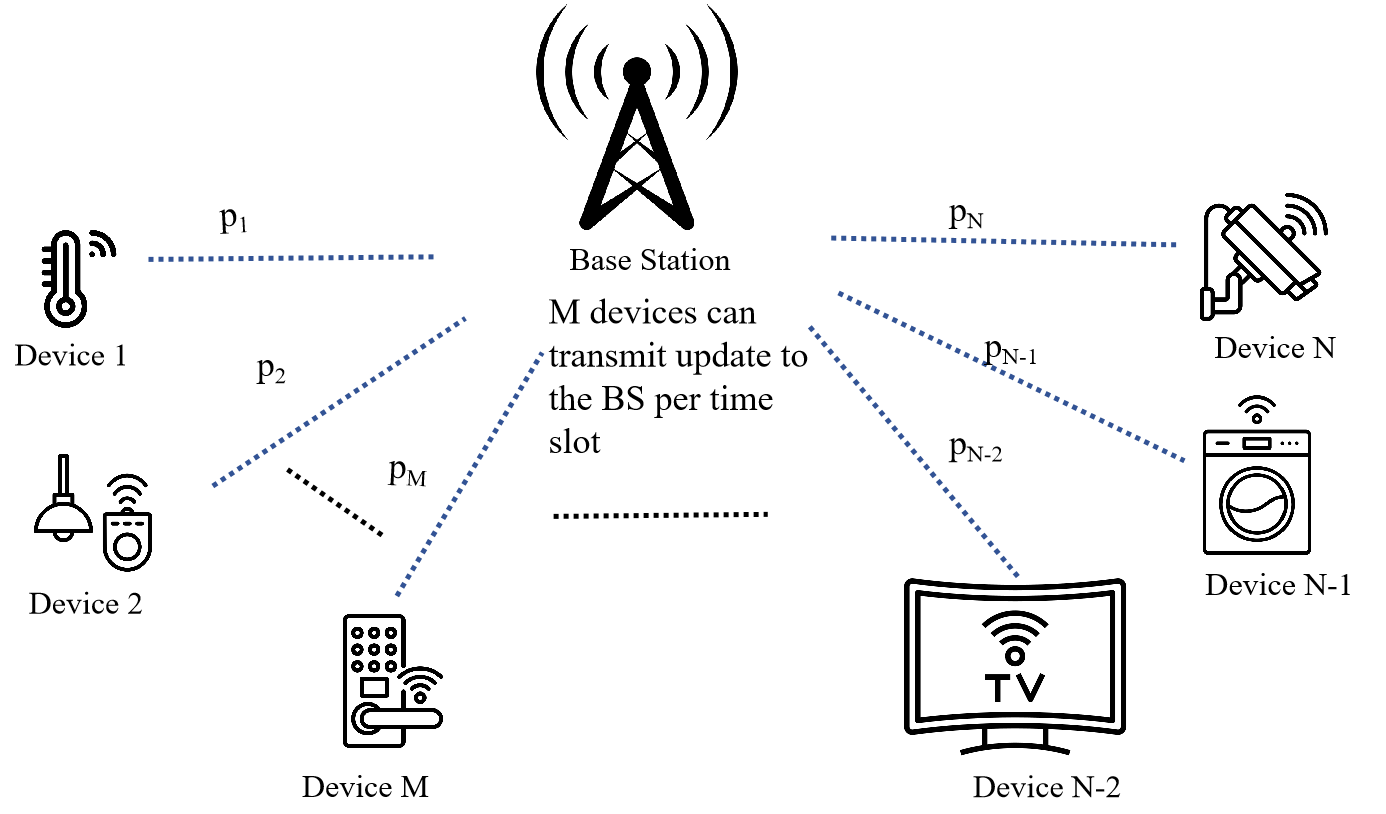}
    \end{center}
    \caption{Smart home system}
    \label{fig:smarthome}
\end{figure}
In this system, multiple smart devices (e.g., smart bulbs, thermostats, security cameras) need to transmit status update information to a central control unit, i.e., the base station (BS), for processing. Time is divided into discrete slots, indexed by $t \in \{1, 2, \dots, T\}$. The wireless channel allows up to $M$ devices to transmit simultaneously in each time slot. At each time slot $t$, the BS selects $M$ devices from the set $\{1, 2, \dots, N\}$ for transmission. Each selected device generates a new status update and transmits it to the BS. A packet transmitted by device $i$ is successfully received by the BS with probability $p_i \in (0,1]$, while the probability of transmission failure is $1 - p_i$. Although $p_i$ remains constant over time, it may vary across different devices. We define \( Z_i(t) \) as an indicator function, where \( Z_i(t) = 1 \) denotes a successful transmission from device \( i \) to the BS at time \( t \), and \( Z_i(t) = 0 \) otherwise.
 
The performance of each device $i$ depends on both the quantity and timeliness of its status updates. Therefore, we evaluate the performance of a device based on its throughput and Age of Information. The throughput of a device is defined as the long-term average number of successful packet deliveries and is given by $ \liminf\limits_{T \to \infty} \frac{\sum_{t=1}^{T} Z_i(t)}{T} $.
In our model, the AoI of a device $i$ at time $t$, denoted by $a_i(t)$, represents the time elapsed since the last successful packet delivery. It evolves according to the following recursive expression:
\begin{equation}
    \label{eq:aoi_update}
    a_i(t+1) =
    \begin{cases} 
        1, & \text{if } Z_i(t) = 1, \\
        a_i(t) + 1, & \text{otherwise}.
    \end{cases}
\end{equation}
The long-term average AoI of device \( i \) is then defined as\\ 
 $\limsup\limits_{T \to \infty} \frac{\sum_{t=1}^{T} a_i(t)}{T}.$

The objective of this paper is to enable a fine-grained trade-off of throughput and AoI among devices with heterogeneous channel conditions and distinct performance preferences. To achieve this, we aim to characterize the throughput-AoI capacity region of wireless networks, which is defined as follows:
\begin{definition}\label{def:capacity region}
The throughput-AoI capacity region of the network is the set of all feasible vectors of throughput-AoI pairs \( \{(m_i, h_i) \mid 1 \leq i \leq N\} \) such that there exists a scheduling policy that ensures:
\[
\liminf_{T \to \infty} \frac{ \sum_{t=1}^T Z_i(t)}{T}\geq m_i ,  
\]
\[
\limsup\limits_{T \to \infty}\frac{\sum_{t=1}^{T} a_i(t)}{T} \leq h_i .
\]
\end{definition}
To illustrate the utility of the throughput-AoI capacity region, we present the following four example applications. To the best of our knowledge, many of these problems remain open.
\begin{Example}\label{example1}[AoI minimization with hard throughput constraints]\\
This problem, introduced by Kadota et al.~\cite{modiano}, considers that each device \( i \) requires a minimum throughput of \( q_i \). The objective is to minimize the total AoI in the system while meeting the minimum throughput requirements. Using our model, the problem can be formulated as the following optimization problem:
\begin{align}
 \min&\sum_{i=1}^N h_i,\\ 
\mbox{s.t. } & m_i \geq q_i, \forall i,\\
&\{(m_i, h_i)|1\leq i\leq N\}\mbox{ is feasible.}
\end{align}
\end{Example}
\begin{Example}\label{example2}[Cost minimization with soft throughput constraints.]\\
Solving the problem in Kadota et al.~\cite{modiano} can lead to undesirable performance when the throughput constraints \(q_i\) are either too strict or infeasible. We consider an extension of the problem studied by Kadota et al.~\cite{modiano}, 
where throughput constraints can be violated, but each violation incurs a penalty 
\( C(v_i) \), with \( v_i \) denoting the amount of violation for device \( i \). 
Our objective is to minimize the sum of the total penalty and the total AoI across 
the entire network. Using our model, this problem can be formulated as the following 
optimization problem:

\begin{align}
\min &\sum_{i=1}^N C(q_i-m_i) + h_i\\
\mbox{s.t. } &\{(m_i, h_i)|1\leq i\leq N\}\mbox{ is feasible.}
\end{align}
\end{Example}
\begin{Example}\label{example3}[Proportional fairness in both throughput and AoI]\\
Proportional fairness is a popular objective in network optimization due to its ability to achieve both high spectrum efficiency and fairness among heterogeneous users. Almost all existing work on proportional fairness in wireless scheduling only focuses on proportional fairness in terms of throughput. In this setting, the objective is to maximize $\sum_i \log m_i$. By characterizing the throughput-AoI capacity region, we can extend proportional fairness to account for both throughput and AoI. For example, we can define the utility of a device $i$ as $\log m_i + \log \frac{1}{h_i}=\log m_i - \log h_i$. Achieving proportional fairness in this setting is equivalent to solving the following optimization problem:
\begin{align}
\max &\sum_{i=1}^N (\log m_i - \log h_i)\\
\mbox{s.t. } &\{(m_i, h_i)|1\leq i\leq N\}\mbox{ is feasible.}
\end{align}
\end{Example}
\begin{Example}\label{example4}[Admission control for devices with AoI constraints]\\
Consider the scenario where each device $i$ has a hard AoI constraint $e_i$ and requires that $h_i\leq e_i$. When a new device arrives in the system, the BS needs to determine whether it is possible to satisfy the AoI constraints of all devices. By using the throughput-AoI capacity region, this admission control problem reduces to determining whether there exist $m_i$ such that the vector of throughput-AoI pairs $\{(m_i, e_i)|1\leq i\leq N\}$ is feasible.
\end{Example}

\section{SECOND-ORDER APPROXIMATION and
System-Wide Temporal Variance}\label{section:secorderappro}
Accurately determining the AoI for each device in a dynamic network environment is challenging due to the stochastic nature of wireless transmissions and scheduling policies. To address this challenge, we employ \emph{second-order analysis}~\cite{daojingsecondorder}. Specifically, we assume that the process $\{Z_i(1), Z_i(2), \dots\}$ evolves as a \emph{positive recurrent Markov process} to characterize the throughput-AoI capacity region in unreliable wireless networks.
In second-order analysis, each random process$\{Z_i(t)\}$ is characterized by its \emph{mean} and \emph{temporal variance}.
The mean of $\{Z_i(t)\}$ is defined as:
\begin{equation}
    \label{eq:mean}
    \mu_i = \lim_{T \to \infty} \frac{\sum_{t=1}^T Z_i(t)}{T}.
\end{equation}

The temporal variance of $\{Z_i(t)\}$ is defined as:
\begin{equation}
    \label{eq:temporal variance}
    \sigma_i^2 = \mathbb{E} \left[\left(\lim_{T \to \infty} \frac{\sum_{t=1}^T Z_i(t) - T\mu_i}{\sqrt{T}} \right)^2\right].
\end{equation}

The existence of the mean and temporal variance is guaranteed by the Markov Law of Large Numbers and Central Limit Theorem.

Previous studies have shown that AoI can be well-approximated as a function of $\mu_i$ and $\sigma_i^2$\cite{daojingsecondorder}:
\begin{equation}
    \label{eq:aoi_appro}
  \limsup\limits_{T \to \infty} \frac{\sum_{t=1}^{T} a_i(t)}{T}\approx\frac{1}{2} \left(\frac{\sigma_i^2}{\mu_i^2} + \frac{1}{\mu_i} \right) + \frac{1}{2}.
\end{equation}

Assuming that the approximation is exact, then the throughput-AoI capacity region can be characterized by determining the set of all feasible \(\{(\mu_i, \sigma_i^2) \mid 1 \leq i \leq N\}\) under network constraints.

However, it is difficult to characterize the variance $\sigma_i^2$ of a device $i$ directly since it depends on both the channel quality $p_i$ and the employed scheduling policy. To overcome this challenge, we introduce a \emph{projected process}, which enables the analysis of \emph{system-wide variance}, a quantity that exhibits more structured behavior and remains invariant to scheduling policies~\cite{qostheory}. By studying the system-wide variance, we can later establish its relationship with the variance of individual devices.

We define the projected process \( X(t) \) as follows:
\begin{equation}
    \label{eq:martingale}
    X(t) := Mt - \sum_{\tau=1}^t \sum_{i=1}^N \frac{Z_i(\tau)}{p_i}.
\end{equation}

At each time slot \( t+1 \), the BS selects a set \( S(t+1) = \{i_1, i_2, \ldots, i_M\} \) of \( M \) devices for transmission. For each \( i \in S(t+1) \), the random variable \( Z_i(t+1) \) follows a Bernoulli distribution with mean \( p_i \). For each $i\notin S(t+1)$, $Z_i(t+1)=0$. Hence, we have:
\begin{align}\nonumber
\mathbb{E}[X(t+1)\mid{H}^t] &= X(t)+M-\sum_{i\in S(t+1)} \left( \frac{1}{p_{i_k}} \cdot p_{i_k} + 0 \cdot (1 - p_{i_k}) \right) \\ 
&= X(t)  
\end{align}
for any $X(t)$ and $S(t+1)$, where ${H}^t$ is the system history up to time $t$. Thus, the projected process $X(t)$ is a \emph{Martingale}.

We first study the mean of $X(t)$. Let $r_i(t)$ be the probability of $i\in S(t)$ and let $\bar{r}_i:=\lim_{T \to \infty} \sum_t r_i(t)/T$, which is the long-term average portion of time that device $i$ is selected for transmission. Previous work \cite{qostheory} has shown that, via the martingale stability theorem, \begin{equation}
\mu_i=p_i\bar{r}_i, \label{eq:r-mu relation}
\end{equation}
and hence $\lim_{T \to \infty} X(T)/T=0$.

Next, we study the variance of $X(t)$, defined as \\
$\sigma^2:=\lim_{T \to \infty} E[X^2(T)/T]$. First, we note that:
\begin{equation}
    \label{eq:sysvariance}
    \begin{aligned}
        \sigma_t^2 &:=E\left[ \left( X(t+1) - X(t) \right)^2 \,\middle|\, H^t \right]=\sum_{i\in S(t+1)}\frac{E[Z^2_i(t+1)]-(E[Z_i(t+1)])^2}{p_i^2}\\
        &=\sum_{i\in S(t+1)}(\frac{1}{p_i}-1).
    \end{aligned}
\end{equation}

By the martingale central limit theorem, we have the long-term system variance:
\begin{align}\label{eq:longtermvariance}
\sigma^2=\sum_{i=1}^N\bar{r}_i(\frac{1}{p_i}-1)=\sum_{i=1}^N\frac{\mu_i}{p_i}(\frac{1}{p_i}-1).
\end{align}

An important property of the projected process $X(t)$ is that: its mean, $\lim_{T \to \infty} X(T)/T$, is a constant 0 under any scheduling policy and its temporal variance $\sigma^2$, is also a constant among all policies that provide the same throughput $\mu_i$. In the following analysis, we will leverage this result to establish the relationship between individual device variance and the system-wide variance.
\section{An outer bound of the Capacity Region}\label{section:outerbound}
In this section, we derive a necessary condition for a given vector of throughput-AoI pairs \(\{(m_i, h_i) \mid 1 \leq i \leq N\}\) to be in the throughput-AoI capacity region.

\begin{theorem}\label{theorem:neccssary condition}
If a vector of throughput-AoI pairs \(\{(m_i, h_i) \mid 1 \leq i \leq N\}\) is in the throughput-AoI capacity region, then there must exist vectors \(\{(\mu_i, \sigma_i^2) \mid 1 \leq i \leq N\}\) that satisfy the following conditions:
\begin{align}
&\mu_i \geq m_i, \forall i,\label{eq:mean1}\\
&\frac{1}{2} \left( \frac{\sigma_i^2}{\mu_i^2} + \frac{1}{\mu_i} \right) + \frac{1}{2}\leq h_i, \forall i,\label{eq:aoi}\\
&\sum_{i=1}^{N} \frac{\mu_i}{p_i} = M,\label{eq:somean}\\
&0 \leq \frac{\mu_i}{p_i} \leq 1,\label{eq:onetime}\\
&\sum_{i=1}^{N} \sqrt{\frac{\sigma_i^2}{p_i^2}} \geq \sqrt{\sigma^2} = \sqrt{\sum_{i=1}^{N} \frac{\mu_i}{p_i} \left( \frac{1}{p_i} - 1 \right)}, \label{eq:sovariance}
\end{align}
\end{theorem}

\begin{proof}\label{proof:neccssary condition}
Consider a scheduling policy that achieves the vector of throughput-AoI pairs $\{(m_i, h_i) \}$. Let $\mu_i$ and $\sigma_i^2$ be the mean and temporal variance of $Z_i(t)$ under this scheduling policy. We will show that the vector $\{(m_i, h_i) \}$ satisfies Eq.~\eqref{eq:mean1} - Eq.~\eqref{eq:sovariance}.

Eq.~\eqref{eq:mean1}and Eq.~\eqref{eq:aoi} are obviously necessary as they state that the empirical throughput needs to be at least $m_i$, and the empirical AoI needs to be at most $h_i$.

We then establish condition Eq. ~\eqref{eq:somean} and Eq. ~\eqref{eq:onetime}. By Eq. (\ref{eq:r-mu relation}),  We have $\sum_i\frac{\mu_i}{p_i}=\sum_i \bar{r}_i$. Since the BS schedules $M$ devices for transmission in each time slot, we have $\sum_i\bar{r}_i=M$. Hence, Eq. ~\eqref{eq:somean} holds. Also, it is obvious that Eq. ~\eqref{eq:onetime} holds.

Next, we establish condition Eq.~\eqref{eq:sovariance}. 
We begin by defining the random variable \(\hat{Z_i}\) as:
\begin{equation}
    \hat{Z_i} := \lim_{T \to \infty} \frac{\sum_{t=1}^T Z_i(t) - T\mu_i}{\sqrt{T}}.
    \label{eq:z_hat}
\end{equation}

Let  \(\hat{X}\) be defined as:
\begin{equation}
    \hat{X} := \lim_{T \to \infty} \frac{MT - \sum_{t=1}^T \sum_{i=1}^N \frac{Z_i(t)}{p_i}}{\sqrt{T}}=-\sum_{i=1}^{N}\frac{\hat{Z_i}}{p_i}.
    \label{eq:x_hat}
\end{equation}

Then:
\begin{align}
    &\left( \sum_{i=1}^N \sqrt{\frac{\sigma_i^2}{p_i^2}} \right)^2 
    = \left( \sum_{i=1}^N \sqrt{\mathbb{E}\left[\frac{\hat{Z_i}^2}{p_i^2}\right]} \right)^2 (\because Eq.~\eqref{eq:temporal variance})\nonumber \\
    =& \sum_{i=1}^N \mathbb{E}\left[\frac{\hat{Z_i}^2}{p_i^2}\right] 
    + 2 \sum_{i \neq j} \sqrt{\mathbb{E}\left[\frac{\hat{Z_i}^2}{p_i^2}\right] \mathbb{E}\left[\frac{\hat{Z_j}^2}{p_j^2}\right]} && \nonumber \\
    \geq& \sum_{i=1}^N \mathbb{E}\left[\frac{\hat{Z_i}^2}{p_i^2}\right] 
    + 2 \sum_{i \neq j} \mathbb{E}\left[\frac{\hat{Z_i}}{p_i} \cdot \frac{\hat{Z_j}}{p_j}\right] 
    (\because \text{Cauchy-Schwarz inequality}) && \nonumber \\
    =& \mathbb{E} \left[ \left( \sum_{i=1}^N \frac{\hat{Z_i}}{p_i} \right)^2 \right] = \mathbb{E}[\hat{X}^2]= \sigma^2.(\because Eq.~\eqref{eq:longtermvariance}) && \label{eq:cauchy_schwarz}
\end{align}

Thus, condition Eq.~\eqref{eq:sovariance} holds. This completes the proof.

\end{proof}
Therefore, identifying all parameter pairs \( \{(\mu_i, \sigma_i^2) \mid 1 \leq i \leq N\} \) that satisfy Eqs.~\eqref{eq:mean1}--\eqref{eq:sovariance} yields an outer bound on the capacity region.

\section{A tight inner bound of the Capacity Region}\label{section:innerbound}
In this section, we derive a sufficient condition for a given vector of throughput-AoI pairs \( \{(m_i, h_i) \mid 1 \leq i \leq N\} \) to belong to the throughput-AoI capacity region. Furthermore, we introduce a scheduling policy that achieves the desired throughputs and AoIs, provided they satisfy the sufficient condition.

\begin{theorem}\label{theorem:sufficient condition}
A vector of throughput-AoI pairs \( \{(m_i, h_i) \mid 1 \leq i \leq N\} \) is in the throughput-AoI capacity region if there exists a vector \( \{(\mu_i, \sigma^2_i) \mid 1 \leq i \leq N\} \) that satisfies the following conditions.
\begin{align}
\nonumber
&Eq.~\eqref{eq:mean1}-Eq.~\eqref{eq:somean}\\ 
&0<\frac{\mu_i}{p_i} < 1, \forall i, \label{eq:eachless1}\\
&\sum_{i=1}^{N} \sqrt{\frac{\sigma_i^2}{p_i^2}} = \sqrt{\sigma^2} = \sqrt{\sum_{i=1}^{N} \frac{\mu_i}{p_i} \left( \frac{1}{p_i} - 1 \right)}.\label{eq:sovariance2}
\end{align}
\end{theorem}

Before proving Theorem~\ref{theorem:sufficient condition}, we first discuss its implications. Comparing Theorem~\ref{theorem:neccssary condition} and Theorem~\ref{theorem:sufficient condition}, we observe that their conditions are nearly identical, except that the sufficient condition requires the strict inequalities in Eq.~\eqref{eq:eachless1}, and a strict equality in Eq.~\eqref{eq:sovariance2}. This difference implies that the sufficient condition defines an inner bound that is nearly tight, with differences primarily occurring at the boundary.

To prove Theorem~\ref{theorem:sufficient condition}, we propose a scheduling policy under which all pairs \(\{(\mu_i, \sigma_i^2) \mid 1 \leq i \leq N\}\) satisfy the conditions given by Eq.~\eqref{eq:somean} and Eq.~\eqref{eq:eachless1}-Eq.~\eqref{eq:sovariance2}.

We define the \emph{deficit} of a device \( i \) at time \( t \) as follows:
\begin{equation}\label{eq:d_i}
d_i(t) = \frac{t \mu_i - \sum_{\tau=1}^{t} Z_i(\tau)}{\sqrt{\sigma_i^2}}.
\end{equation}

We note that $t\mu_i$ is the number of deliveries needed for device $i$ by time $t$ to have a mean throughput of $\mu_i$ and $\sum Z_i(\tau)$ is the actual number of deliveries for device $i$ up to time $t$. Hence, $t\mu_i-\sum Z_i(\tau)$ represents how much the actual deliveries for device $i$ deviates from the target mean $\mu_i$. We further normalize this deficit by the target variance $\sqrt{\sigma_i^2}$.

In each time slot, our scheduling policy simply sorts all devices by their deficits and then schedules the $M$ devices with the largest deficits. We name this policy the \emph{variance-weighted deficit} (VWD) policy. Under VWD, \( \{(m_i, h_i) \mid 1 \leq i \leq N\} \) feasible means that there exists\( \{(\mu_i, \sigma^2_i) \mid 1 \leq i \leq N\} \)that satisfy Eq.~\eqref{eq:mean1}-Eq.~\eqref{eq:somean} and Eq.~\eqref{eq:eachless1}-Eq.~\eqref{eq:sovariance2}.

To analyze the performance of the VWD policy, we first define the system-wide deficit as: 
\begin{equation}\label{eq:D_i}
D(t) = \frac{\sum_{i=1}^{N} \sqrt{\frac{\sigma_i^2}{p_i^2}} d_i(t)}
{\sum_{i=1}^{N} \sqrt{\frac{\sigma_i^2}{p_i^2}}}.
\end{equation}

We also define $W_i(t)=d_i(t)-D(t)$ and consider the Lyapunov function $L(t):=\frac{1}{2}\sum_i \sqrt{\frac{\sigma_i^2}{p_i^2}}W_i^2(t)$. 

The VWD policy schedules the device with the largest \(d_i(t-1)\), which is also the device with the largest $W_i(t-1)$. Hence, under the VWD policy, the system can be modeled as a Markov process whose state consists of $W_i(t)$ of all devices. We first show that the Markov process is positive-recurrent.

Let $\Delta d_i(t):=d_i(t)-d_i(t-1)$, $\Delta D(t):= D(t)-D(t-1)$, and $\Delta L(t):= L(t)-L(t-1)$, we can then derive the expected one-step Lyapunov drift as follows:
\begin{align}
    &\Delta(L(t)) := E[L(t) - L(t - 1) \mid H^{t-1}]  \nonumber \\
    =&E\Bigg[\frac{1}{2} \sum_{i=1}^N  \sqrt{\frac{\sigma_i^2}{p_i^2}}W_i^2(t) - \frac{1}{2} \sum_{i=1}^N \sqrt{\frac{\sigma_i^2}{p_i^2}}W_i^2(t-1)\;\Bigg| H^{t-1}\Bigg]  \nonumber \\
    =&E\Bigg[\sum_{i=1}^N \sqrt{\frac{\sigma_i^2}{p_i^2}}W_i(t-1) \left( {\Delta d_i(t)} - \Delta D(t) \right) \nonumber \\
    &+\frac{1}{2} \sum_{i=1}^N\sqrt{\frac{\sigma_i^2}{p_i^2}}\left( {\Delta d_i(t)} - \Delta D(t) \right)^2 \;\Bigg| H^{t-1}\Bigg]  \nonumber \\
    \leq& B + E\Bigg[\sum_{i=1}^N\sqrt{\frac{\sigma_i^2}{p_i^2}} W_i(t-1) \Delta d_i(t) - \sum_{i=1}^N \sqrt{\frac{\sigma_i^2}{p_i^2}}W_i(t-1) \Delta D(t) \;\Bigg| H^{t-1}\Bigg]  \nonumber \\
    =&B + E\Bigg[\sum_{i=1}^N \sqrt{\frac{\sigma_i^2}{p_i^2}}W_i(t-1) \Delta d_i(t) \;\Bigg| H^{t-1}\Bigg],  \label{eq:delta_L}
\end{align}
where \(B\) is a bounded constant. The last two steps hold because \(\Delta d_i(t)\) and \(\Delta D(t)\) are bounded and \(\sum_{i=1}^N\sqrt{\frac{\sigma_i^2}{p_i^2}}d_i(t-1) = \sum_{i=1}^N \sqrt{\frac{\sigma_i^2}{p_i^2}} D(t-1)\), we can derive that \( E[\Delta D(t)] = 0 \).

\begin{lemma}\label{lemma:positive recurrent}
Assume that the sufficient conditions are satisfied. Then, under the VWD policy, the system-wide Markov process, whose state consists of $W_i(t)$, is positive-recurrent.
\end{lemma}
\begin{proof}\label{proof:positive recurrent}
Without loss of generality, we assume that \( W_1(t) \geq W_2(t) \geq \dots \geq W_N(t) \).
By the design of VWD, the BS schedules devices \( 1 \sim M \) at time \( t \). Hence, 
\begin{equation}
E[\Delta d_i(t)] =
\begin{cases}
\frac{\mu_i-p_i}{\sigma_i}, & \quad \text{for } i \in \{1, 2, \dots, M\}, \\
\frac{\mu_i}{\sigma_i}, & \quad \text{for } i \in \{M+1, \dots, N\}.
\end{cases}
\label{eq:delta_di}
\end{equation}

We then have
\begin{align}
&E \left[ \sum_{i=1}^{N} \sqrt{\frac{\sigma_i^2}{p_i^2}}W_i(t) \Delta d_i(t) \right]  =\sum_{i=1}^{M} W_i(t) ( \frac{\mu_i}{p_i} - 1) 
+ \sum_{i=M+1}^{N} W_i(t) \frac{\mu_i}{p_i} \nonumber \\ 
&= \sum_{i=1}^{M} (W_i(t) - W_{i+1}(t)) \left( \sum_{j=1}^{i} \frac{\mu_j}{p_j} - i \right) \nonumber \\ 
&\quad + \sum_{i=M+1}^{N} (W_i(t) - W_{i+1}(t)) \left( \sum_{j=1}^{i} \frac{\mu_j}{p_j} - M \right) \nonumber 
+ W_N \left( \sum_{i=1}^{N} \frac{\mu_i}{p_i} - M \right) \nonumber \\ 
&= \sum_{i=1}^{M} (W_i(t) - W_{i+1}(t)) \left( \sum_{j=1}^{i} \frac{\mu_j}{p_j} - i \right) \nonumber \\ 
&\quad + \sum_{i=M+1}^{N} (W_i(t) - W_{i+1}(t)) \left( \sum_{j=1}^{i} \frac{\mu_j}{p_j} - M \right). \label{eq:lemma1}
\end{align}

Let 
$\varepsilon = \min \left\{ \min_i \frac{\mu_i}{p_i}, 1 - \max_i \frac{\mu_i}{p_i} \right\}.$ We have $\epsilon >0$ due to Eq.~\eqref{eq:eachless1}. 
Then, for any \( i \in \{1, 2, \dots, M\} \),
\begin{equation}
\sum_{j=1}^{M} \frac{\mu_j}{p_j} - i < -\varepsilon,
\label{eq:ineq1}
\end{equation}
and, for any \( i \in \{M+1, \dots, N-1\} \),
\begin{equation}
\sum_{j=1}^{M} \frac{\mu_j}{p_j} - M = \sum_{j=M+1}^{N} \frac{\mu_j}{p_j} < -\varepsilon.
\label{eq:ineq2}
\end{equation}

Therefore, combining Eq. ~\eqref{eq:lemma1}, Eq. ~\eqref{eq:ineq1} and Eq. ~\eqref{eq:ineq2}, we can get:
\begin{align}   
&E \left[ \sum_{i} \sqrt{\frac{\sigma_i^2}{p_i^2}}W_i(t) \Delta d_i(t) \right]< -\varepsilon \left( W_1(t) - W_N(t) \right)< -\varepsilon W_1(t)
\label{eq:expectation_ineq}
\end{align}
The last step holds because $\sum_{i=1}^{N}\sqrt{\frac{\sigma_i^2}{p_i^2}} W_i(t)=0, W_N(t)$ must be negative.
Thus, if
\[
\max_{i} W_i(t) > \frac{B}{\varepsilon},
\]
then \( E[\Delta L] < 0 \).

Since $W_i(t)$ of each device evolves according to a Markov process with a finite state space, the system dynamics can be analyzed within this framework. Obviously, all system states satisfying 
$\max_{i} W_i(t) \leq \frac{B}{\varepsilon}$ belong to a finite set of states. By the Foster-Lyapunov Theorem, it follows that the system-wide Markov process is positive-recurrent.
\end{proof}
We now show that the VWD policy achieves the desired throughputs and AoIs that satisfy the sufficient conditions, thereby establishing Theorem ~\ref{theorem:sufficient condition}.\\
\begin{theorem}\label{theorem:verification}
Assume that the sufficient conditions are satisfied. Then, under the VWD policy, the following holds:
\begin{align}
&\lim_{T \to \infty} \frac{\sum_{t=1}^T Z_i(t)}{T} = \mu_i, \forall i,\\
&\mathbb{E}\left[\left(\lim_{T \to \infty} \frac{\sum_{t=1}^T Z_i(t) - T \mu_i}{\sqrt{T}}\right)^2\right] = \sigma_i^2, \forall i.
\end{align}
\end{theorem}
\begin{proof}\label{proof:verification}
Since the system-wide Markov process is positive recurrent under the VWD policy, we have:
\begin{equation}\label{eq:T}
\lim_{T \to \infty} \frac{d_i(T)- D(T)}{T} \to 0, \quad \forall i,
\end{equation}
\begin{equation}\label{eq:sqrt_T}
\lim_{T \to \infty} \frac{d_i(T) - D(T)}{\sqrt{T}} \to 0, \quad \forall i. 
\end{equation}

First, we show that \(\lim_{T \to \infty} \frac{\sum_{t=1}^T Z_i(t)}{T} = \mu_i\), \(\forall i\). Recall that \(d_i(t) = \frac{t \mu_i - \sum_{\tau=1}^{t} Z_i(\tau)}{\sqrt{\sigma_i^2}}.\) and \(D(t) = \frac{\sum_{i=1}^N \sqrt{\frac{\sigma_i^2}{p_i^2}} d_i(t)}{\sum_{i=1}^N \sqrt{\frac{\sigma_i^2}{p_i^2}} }\). By Eq.~\eqref{eq:somean}, we have:
\begin{align}
&\lim_{T \to \infty} \frac{D(T)}{T} \nonumber
= \lim_{T \to \infty} 
\frac{\sum_{i=1}^N\sqrt{\frac{1}{p_i^2}}( T {\mu_i} - \sum_{t=1}^T {Z_i(t)})}
{T \sum_{i=1}^N \sqrt{\frac{\sigma_i^2}{p_i^2}}} \\ \nonumber
&= \lim_{T \to \infty} 
\frac{ MT - \sum_{i=1}^N \sum_{t=1}^T \frac{Z_i(t)}{p_i}}
{T \sum_{i=1}^N \sqrt{\frac{\sigma_i^2}{p_i^2}}} \\
&= \lim_{T \to \infty} \frac{M - \frac{1}{T} \sum_{i=1}^N \sum_{t=1}^T \frac{Z_i(t)}{p_i}}
{\sum_{i=1}^N \sqrt{\frac{\sigma_i^2}{p_i^2}}}\label{eq:mean3}
\end{align}

According to the Law of Large Numbers, we have:\\
\(\frac{1}{T} \sum_{i=1}^N \sum_{t=1}^T \frac{Z_i(t)}{p_i} \to M\) as \(T \to \infty\).

Therefore, since \(\lim_{T \to \infty} \frac{D(T)}{T} = 0\), according to Eq. ~\eqref{eq:T}, we obtain:
\[
\lim_{T \to \infty} \frac{d_i(T)}{T} = \frac{\mu_i}{\sqrt\sigma_i^2} - \lim_{T \to \infty} \frac{\sum_{t=1}^T \frac{Z_i(t)}{\sqrt\sigma_i^2}}{T} = 0, \quad \text{for all } i.
\]

Rearranging the terms, we conclude:
\(\lim_{T \to \infty} \frac{\sum_{t=1}^T Z_i(t)}{T} = \mu_i\).

Next, we show that $\mathbb{E} [ ( \lim_{T \to \infty} \frac{\sum_{t=1}^T Z_i(t) - T \mu_i }{\sqrt{T}} )^2 ] = \sigma_i^2, \quad \forall i$.
\begin{align}
&\lim_{T \to \infty} \frac{D(T)}{\sqrt{T}} =\lim_{T \to \infty} \nonumber
\frac{\sum_{i=1}^N\sqrt{\frac{1}{p_i^2}}( T {\mu_i} - \sum_{t=1}^T {Z_i(t)})}
{\sqrt T \sum_{i=1}^N \sqrt{\frac{\sigma_i^2}{p_i^2}}} \\\nonumber
&= \lim_{T \to \infty} 
\frac{MT - \sum_{i=1}^N \sum_{t=1}^T \frac{Z_i(t)}{p_i}}
{\sqrt{T} \sum_{i=1}^N \sqrt{\frac{\sigma_i^2}{p_i^2}}} \\
&= \lim_{T \to \infty} 
\frac{X(T)}{\sqrt{T} \sum_{i=1}^N \sqrt{\frac{\sigma_i^2}{p_i^2}}}.\label{eq:variance3}
\end{align}

According to Martingale Central Limit Theorem, $\lim_{T \to \infty} \frac{X(T)}{\sqrt{T}} \sim \mathcal{N}(0,\sum_{i=1}^N \sqrt{\frac{\sigma_i^2}{p_i^2}})$

Then we can get:
$\lim_{T \to \infty} \frac{D(T)}{\sqrt{T}} \sim \mathcal{N}(0,1).$

According to Eq. ~\eqref{eq:sqrt_T}:
\begin{align}
&\mathbb{E}\left[\left(\lim_{T \to \infty} \frac{\sum_{t=1}^T Z_i(t) - T \mu_i}{\sqrt{T}}\right)^2\right] \nonumber
= \mathbb{E}\left[\left(\lim_{T \to \infty} \frac{\sqrt{\sigma_i^2} d_i(T)}{\sqrt{T}}\right)^2\right] \\
&= \sigma_i^2 \mathbb{E}\left[\left(\lim_{T \to \infty} \frac{D(T)}{\sqrt{T}}\right)^2\right] = \sigma_i^2.
\end{align}
\end{proof}
These results trivially imply that the corresponding \( \{(m_i, h_i) \mid 1 \leq i \leq N\} \) in Eq.~\eqref{eq:mean1} and Eq.~\eqref{eq:aoi} are achievable. 

The derivations in Sections~\ref{section:outerbound} and ~\ref{section:innerbound} are based on the approximation in Eq.~\eqref{eq:aoi_appro}. However, in practice, the approximation might not be exact. Assuming that the margin of error for device $i$ is at most $\delta_i$, we can take this error into account by changing Eq.~\eqref{eq:aoi} in the outer bound to:
\[
\frac{1}{2} \left( \frac{\sigma_i^2}{\mu_i^2} + \frac{1}{\mu_i} \right) + \frac{1}{2}-\delta_i \leq h_i,\forall i.
\]
and the inner bound to: 
\[
\frac{1}{2} \left( \frac{\sigma_i^2}{\mu_i^2} + \frac{1}{\mu_i} \right) + \frac{1}{2}+\delta_i \leq h_i,\forall i.
\]


\section{Simulation Results and Analysis}\label{section:simulation}
In this section, we present the simulation results for the proposed VWD scheduler. We apply our policy to four practical and fundamental problems proposed in Section~\ref{section:model}. In the following simulations, each experiment is conducted over \(1,000,000 \times N\) time slots to ensure a comprehensive evaluation of the scheduling policy. The performance metrics are averaged over 1,000 independent traces to obtain statistically meaningful results. In order to evaluate the proposed policy, we also include numerical solutions obtained from solving Example~\ref{example1}, ~\ref{example2}, ~\ref{example3}, ~\ref{example4}, which serve as benchmarks and are referred to as the Theoretical results. 
\subsection{Minimizing Age of Information with hard throughput requirement constraints}\label{simulation1}
In Example~\ref{example1}, we compare the empirical performance of VWD against the Max-Weight scheduling policy~\cite{modiano}, which, to the best of our knowledge, represents a near-optimal approach in practical scenarios.

\textbf{Max-weight policy}: The Max-Weight policy prioritizes the node with the highest weight metric $W_i(t)$, where:
\begin{align}
&W_i(t) = \frac{\alpha_i p_i}{2} h_i(t) [h_i(t) + 2] + V p_i x_i^+(t)\\
&x_i(t+1) = t q_i - \sum_{\tau=1}^{t} Z_i(\tau)\\
&V=N^2
\end{align}\label{eq:maxweight}


We first consider scheduling \(M\) devices in each time slot. Fig.~\ref{fig:N/M change} illustrates the average AoI as \(M\) increases. We consider different ratios \(N/M \in \{3, 5, 10\}\). For the initial set of devices, the channel reliabilities are set as \(p_i = \frac{i}{N}\) and the throughput requirements as \(q_i = \frac{\lambda p_i}{N}\), with \(\lambda = 0.9\). These settings are maintained consistently as \(M\) is scaled up.
\begin{figure*}[t]
    \centering
    \subfigure[$N/M = 3$ ]{
       \includegraphics[width=0.3\linewidth]{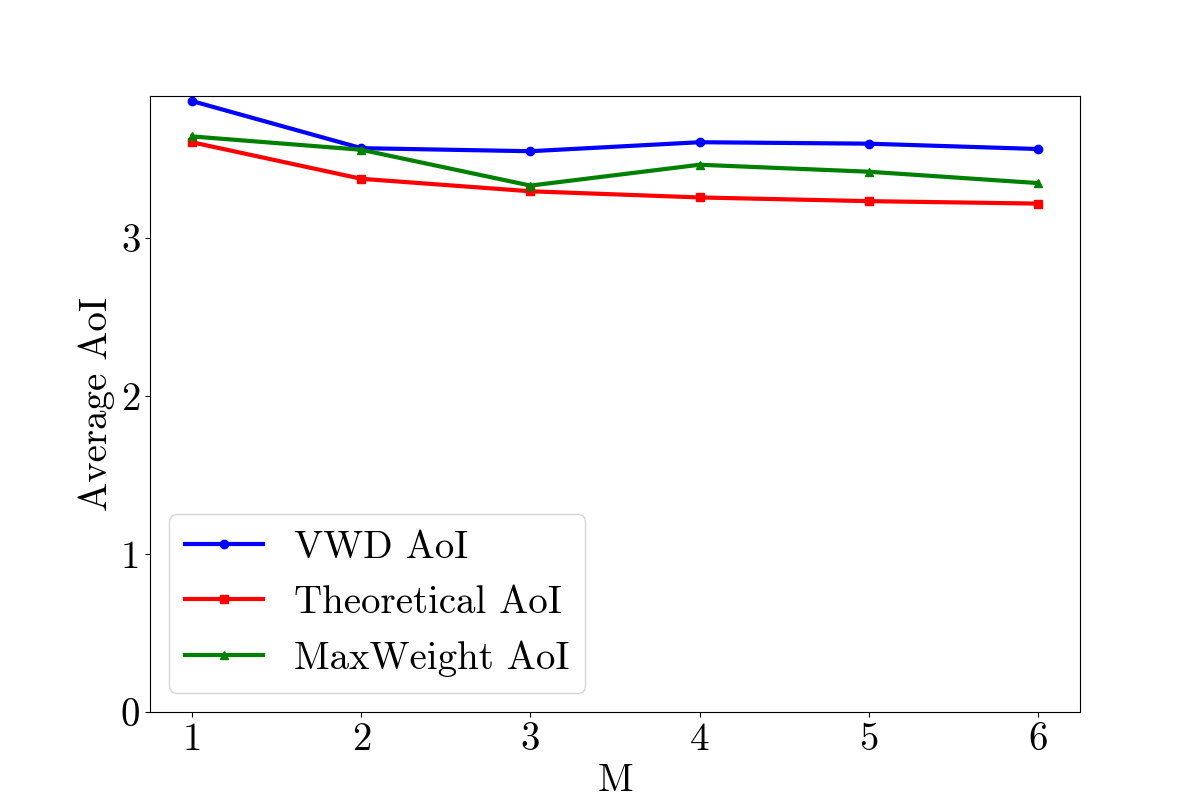}
       \label{fig:N3_1}
    }
    \hfill
    \subfigure[ $N/M = 5$  ]{
       \includegraphics[width=0.3\linewidth]{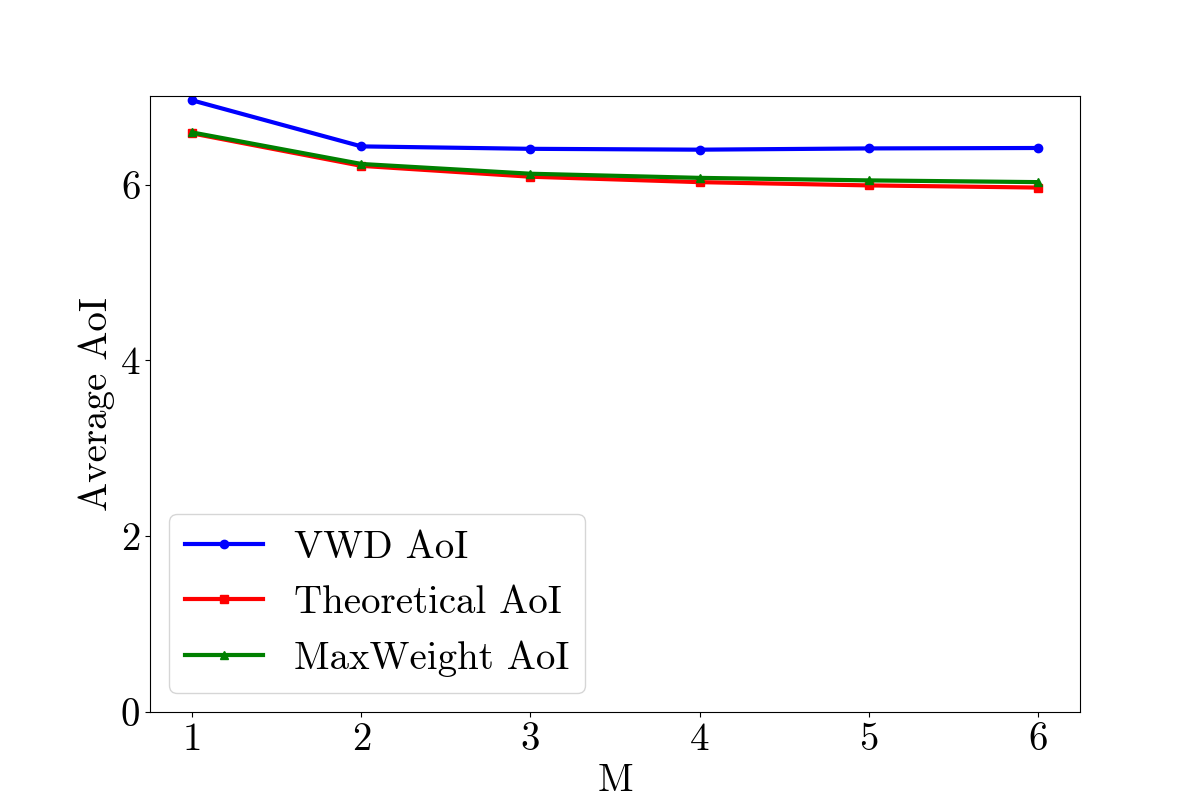}
       \label{fig:N5_1}
    }
    \hfill
    \subfigure[$N/M = 10$  ]{
       \includegraphics[width=0.3\linewidth]{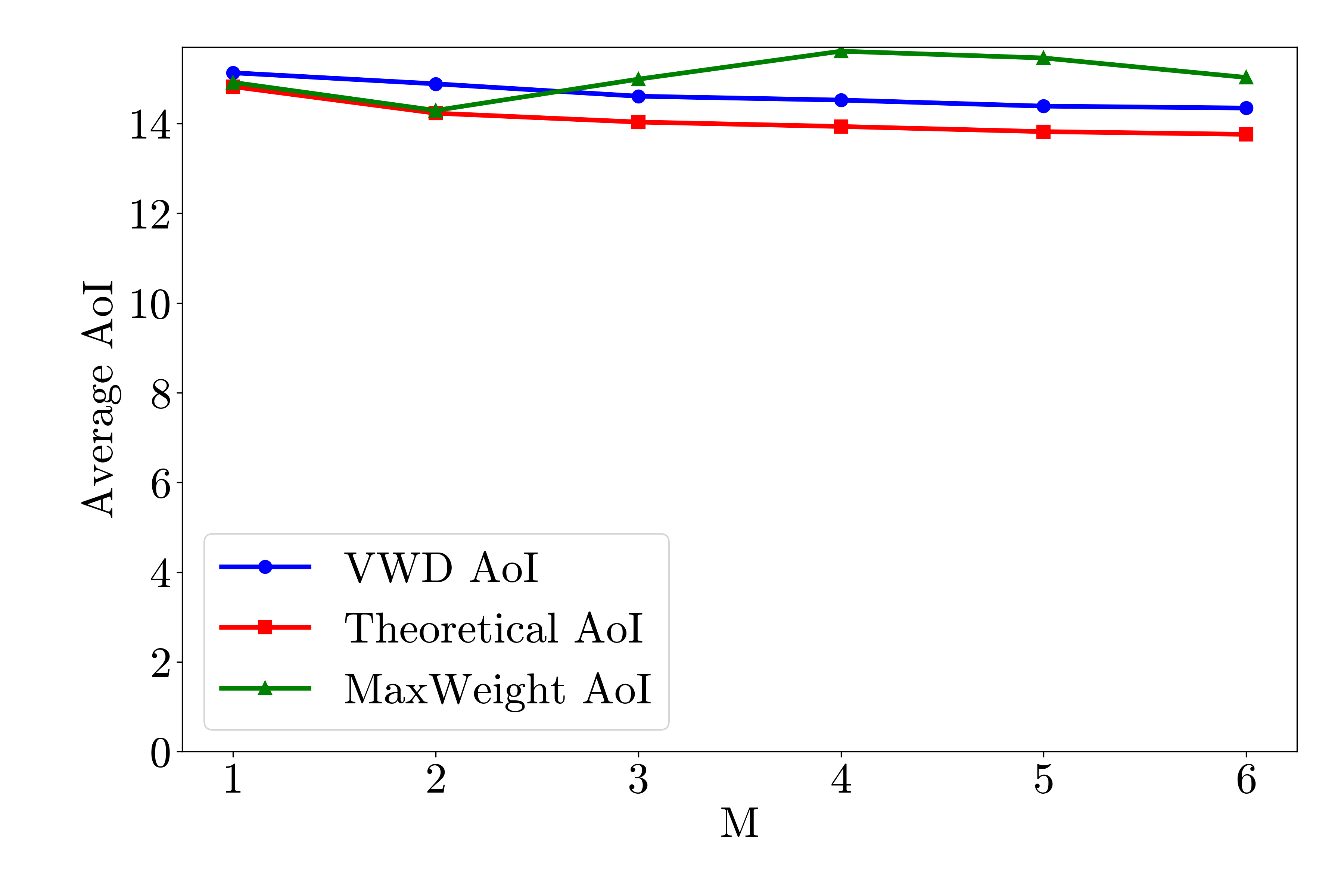}
       \label{fig:N10_1}
    }
    \caption{Averaged Age of Information (AoI) when scale up M.}
    \label{fig:N/M change}
\end{figure*}
Our results indicate that the proposed policy attains performance close to that of the Max-weight Policy, with both closely aligning with our theoretical results. Across most system configurations, the difference between theoretical and VWD performance is small.

We notice that there can be considerable error between theoretical and VWD performance, especially when N/M is small. The approximation in Eq.~\eqref{eq:aoi_appro} is derived by assuming that the inter-delivery time follows an inverse Gaussian distribution, which is a continuous random variable. However, since our system operates in discrete time, the actual inter-delivery time must be an integer. We conjecture that the approximation error shown in Fig.~\ref{fig:N/M change} is mainly caused by the discrepancies between continuous and discrete random variables. To verify this conjecture, 
we conduct simulations under the same settings as in Fig.~\ref{fig:N/M change}, 
comparing the empirical cumulative distribution function (CDF) of inter-delivery times 
with the corresponding fitted inverse Gaussian CDF. 
For each configuration in Fig.~\ref{fig:N/M change}, we focus on the device exhibiting 
the largest approximation error. 
Fig.~\ref{fig:appro error} illustrates the CDF comparisons for different $N/M$ ratios.

\begin{figure*}[t]
    \centering
    \subfigure[$N/M = 3, M=1$ ]{
       \includegraphics[width=0.3\linewidth]{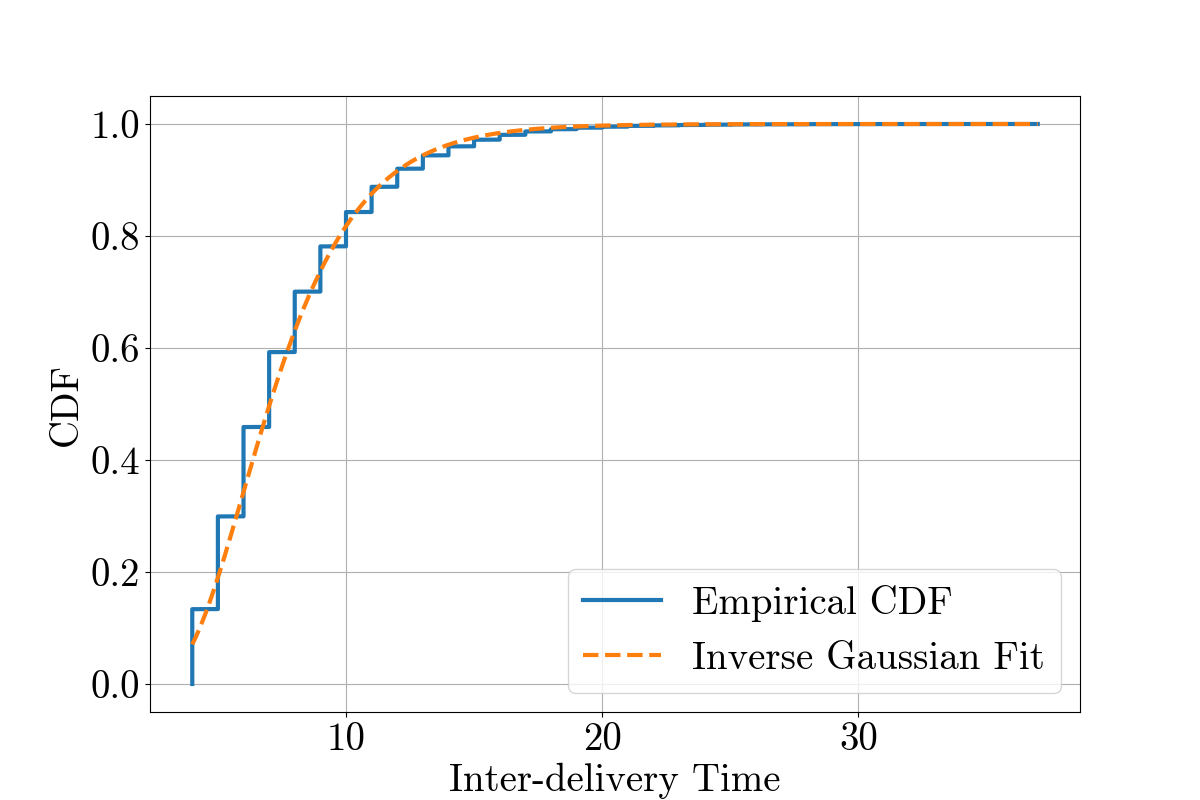}
       \label{fig:N3_2}
    }
    \hfill
    \subfigure[ $N/M = 5, M=1$  ]{
       \includegraphics[width=0.3\linewidth]{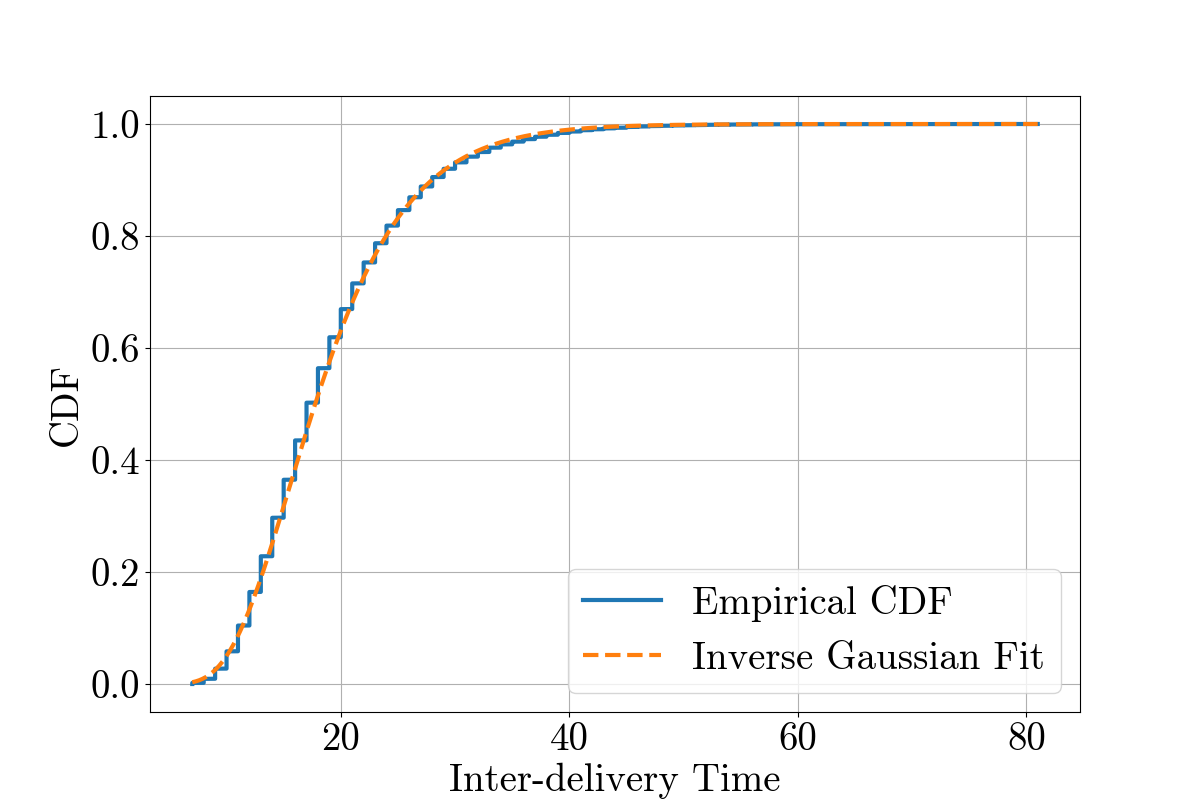}
       \label{fig:N5_2}
    }
    \hfill
    \subfigure[$N/M = 10, M=1$  ]{
       \includegraphics[width=0.3\linewidth]{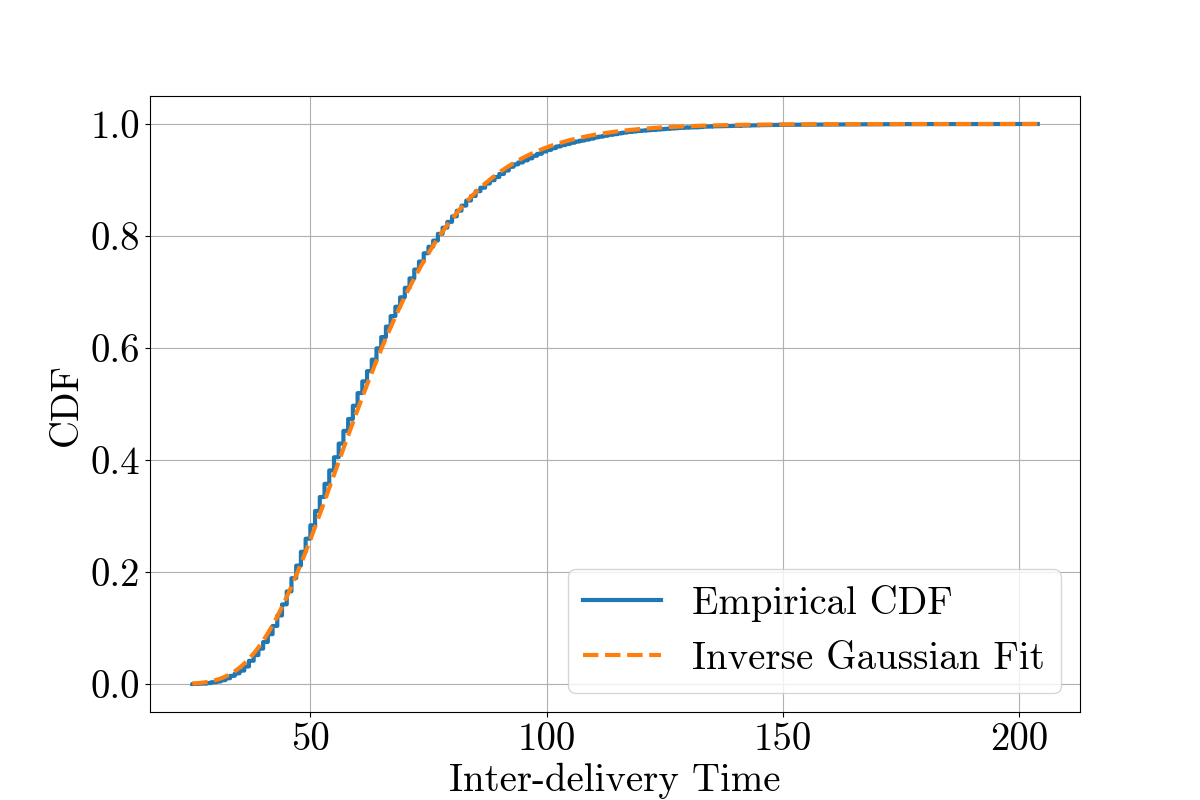}
       \label{fig:N10_2}
    }
    \caption{CDF Comparison between Empirical and Inverse Gaussian as M scale up.}
    \label{fig:appro error}
\end{figure*}
As can be observed in Fig.~\ref{fig:appro error}, the two CDFs are almost identical on all integer points, showing that their differences are indeed mainly due to the continuous-discrete discrepancy. As $N/M$ increases, this discrepancy diminishes, which explains why the difference between theoretical and VWD performance becomes smaller as $N/M$ increases.
\subsection{Cost minimization with soft throughput constraints}
This is the problem described in Example~\ref{example2}, where each device has a throughput requirement of $q_i$. When the throughput requirement is violated, i.e.,$m_i<q_i$, there is a penalty of $C(q_i-m_i)$. We set $C(x)=x^2$ if $x>0$ and $C(x)=0$ if $x\leq 0$. Similar to Section~\ref{simulation1}, we compare the empirical performance of VWD against its theoretical value and the Max-Weight policy. 

We consider two scenarios. In the first scenario, we fix $M=1$. We set $p_i=0.8$ and $q_i=\lambda\gamma_i$, where $\gamma_i = \frac{1.6p_iM}{N}$ if $i \leq N/2$ and $\gamma_i=\frac{0.4p_iM}{N}$ if $i>N/2$. We evaluate three networks with $N=6, 10,$ and $20$. For each network, we evaluate the average cost, defined as the average of the sum of penalty and AoI, across different $\lambda$. Simulation results are shown in Fig.~\ref{fig:f2_1}. Our simulation results show that the empirical performance of VWD is very close to the theoretical value. Moreover, VWD significantly outperforms the Max-Weight policy when $\lambda$ is large. When $\lambda$ is large, it becomes difficult to satisfy the throughput requirements, and policies aiming to satisfy them can result in poor AoI. In contrast, our model allows VWD to suffer a small amount of violation in throughput requirements in exchange for much better AoI. These simulation results demonstrate that achieving a fine-grained trade-off between throughput and AoI can lead to significantly improved system performance.

In the second scenario, we fix the ratio between $N$ and $M$ and evaluate how the average cost changes as $M$ increases. We set $p_i$ and $q_i$ the same as the Section~\ref{simulation1}, and set $\lambda = 0.9$. Simulation results are shown in Fig.~\ref{fig:f2_M}, which shows that VWD is very close to the theoretical value and is much better than the Max-Weight policy.
\begin{figure*}[t]
    \centering
    \subfigure[$N = 6$ devices.]{
       \includegraphics[width=0.3\linewidth]{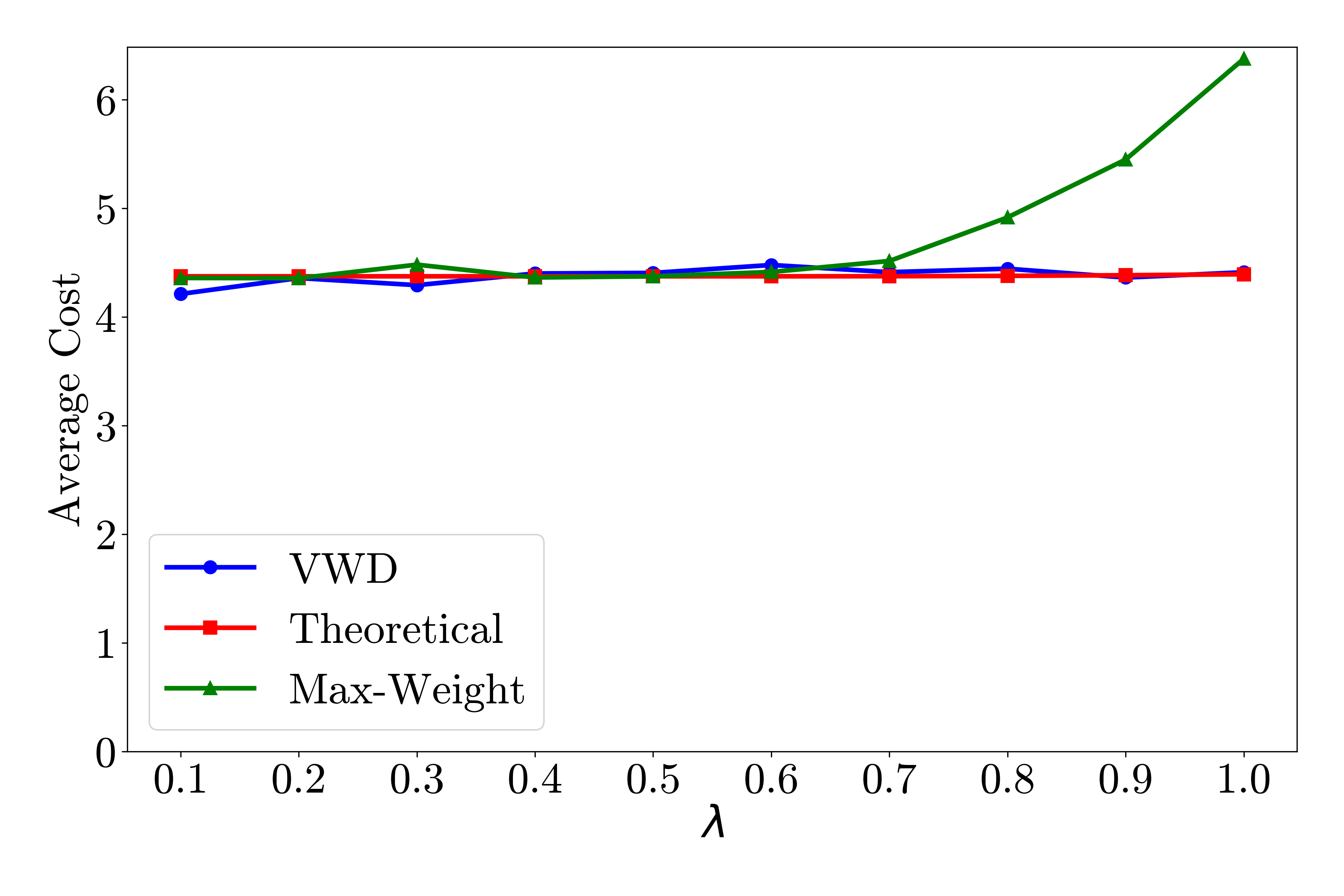}
       \label{fig:N6_2}
    }
    \hfill
    \subfigure[ $N = 10$  devices.]{
       \includegraphics[width=0.3\linewidth]{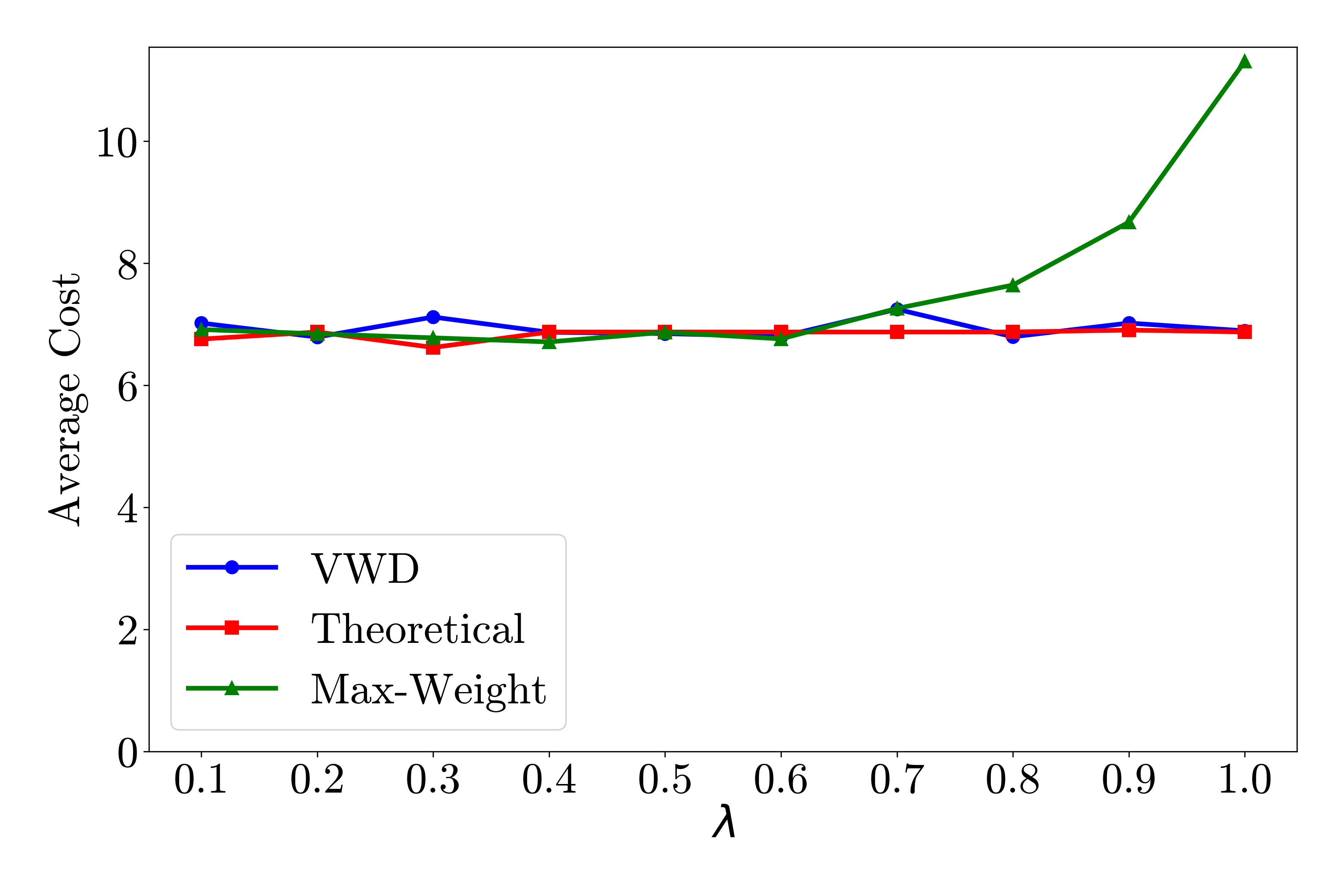}
       \label{fig:N10_2}
    }
    \hfill
    \subfigure[$N = 20$  devices.]{
       \includegraphics[width=0.3\linewidth]{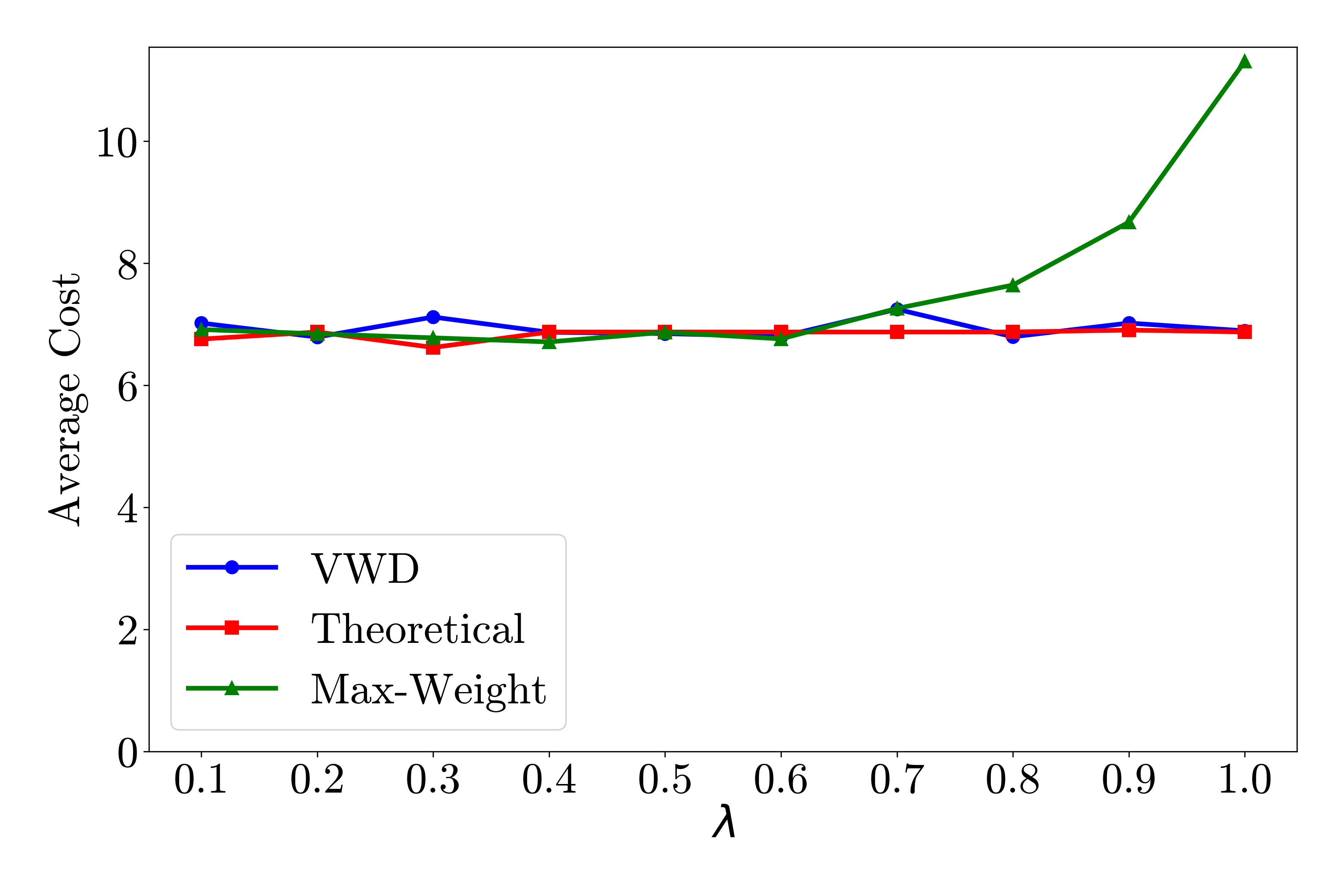}
       \label{fig:N20_2}
    }
    \caption{Average Cost when $\lambda$ changes.}
    \label{fig:f2_1}
\end{figure*}

\begin{figure*}[t]
    \centering
    \subfigure[$N/M= 3$ ]{
       \includegraphics[width=0.3\linewidth]{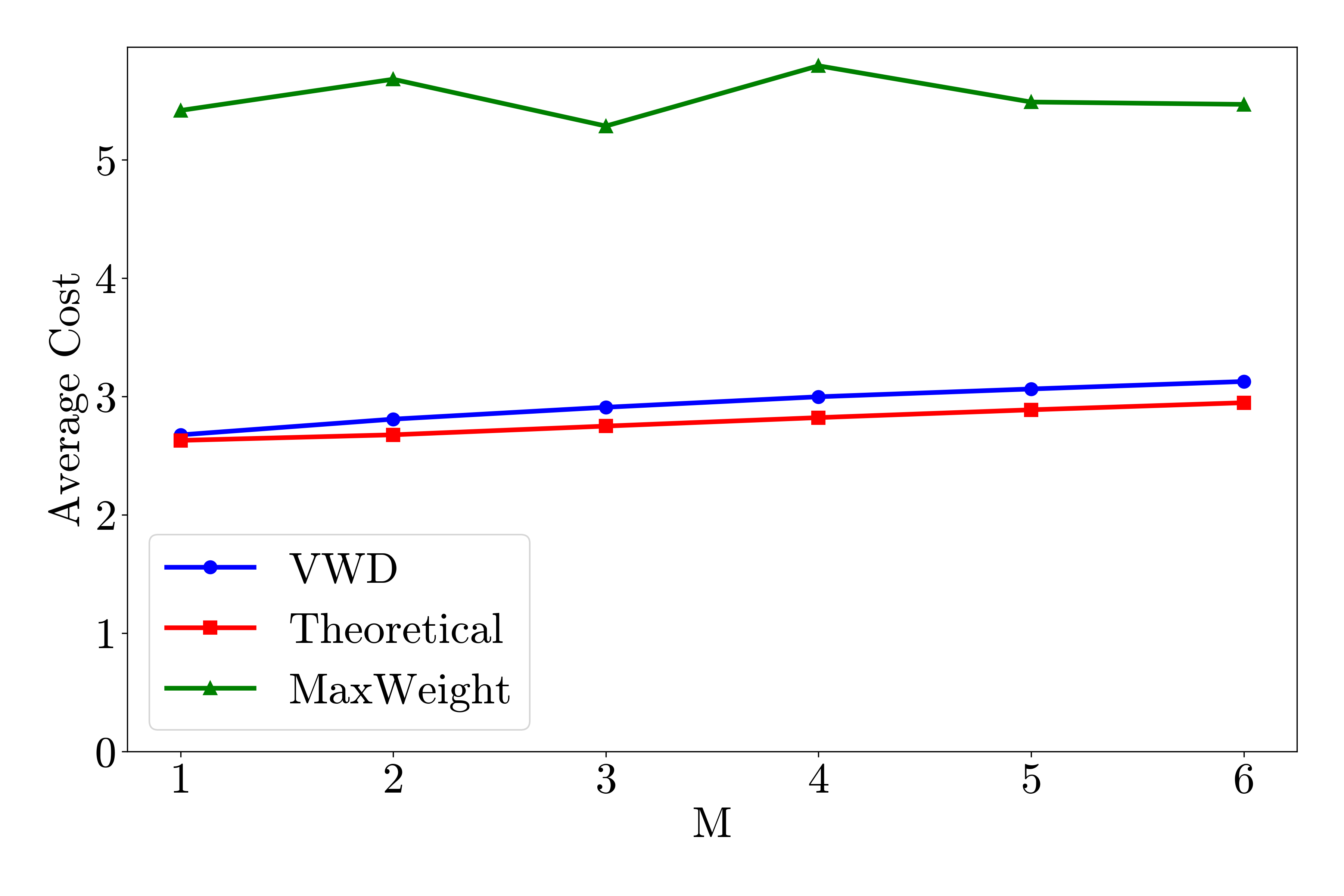}
       \label{f2:M3}
    }
    \hfill
    \subfigure[ $N/M= 5$ ]{
       \includegraphics[width=0.3\linewidth]{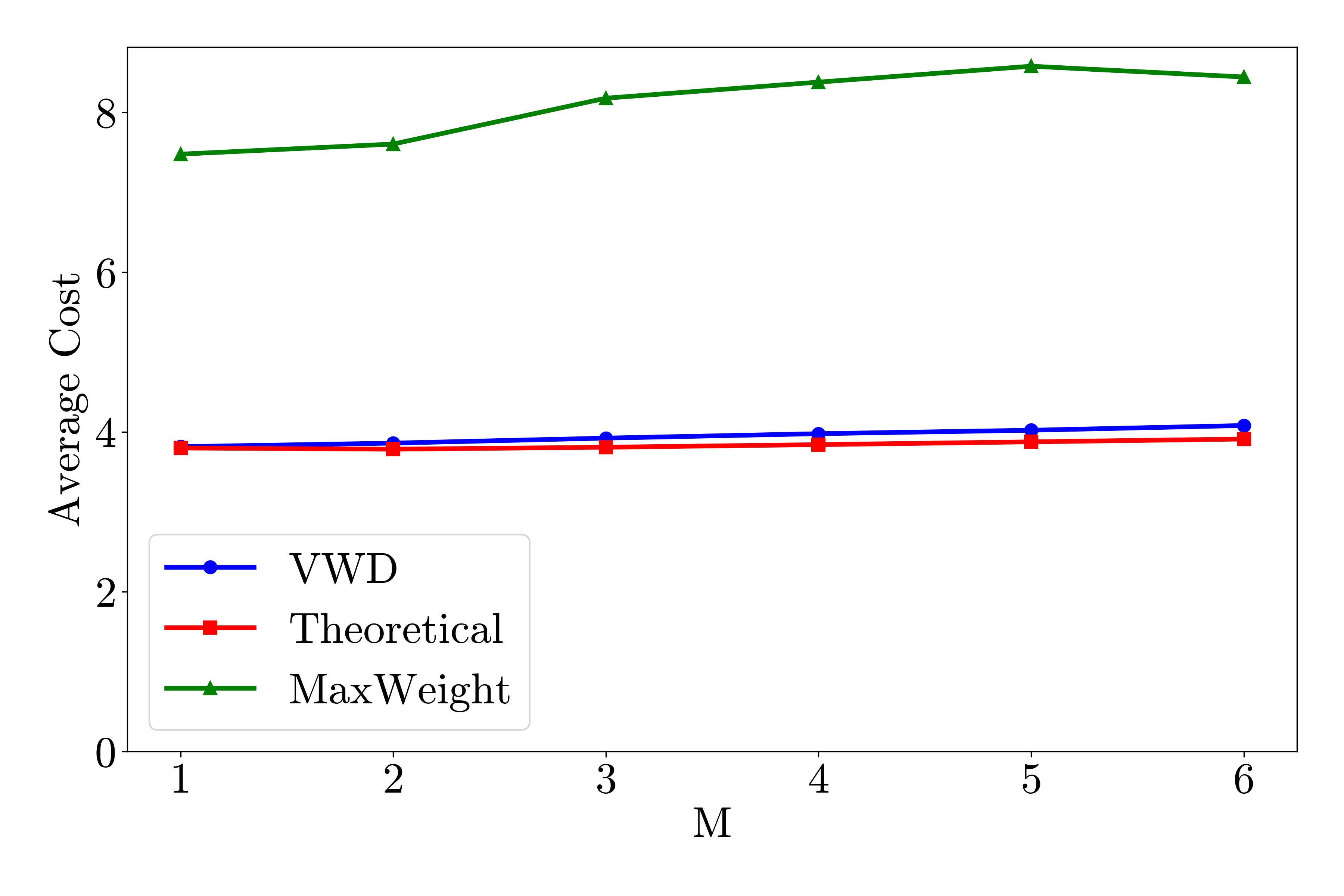}
       \label{f2:M5}
    }
    \hfill
    \subfigure[$N/M= 10$  ]{
       \includegraphics[width=0.3\linewidth]{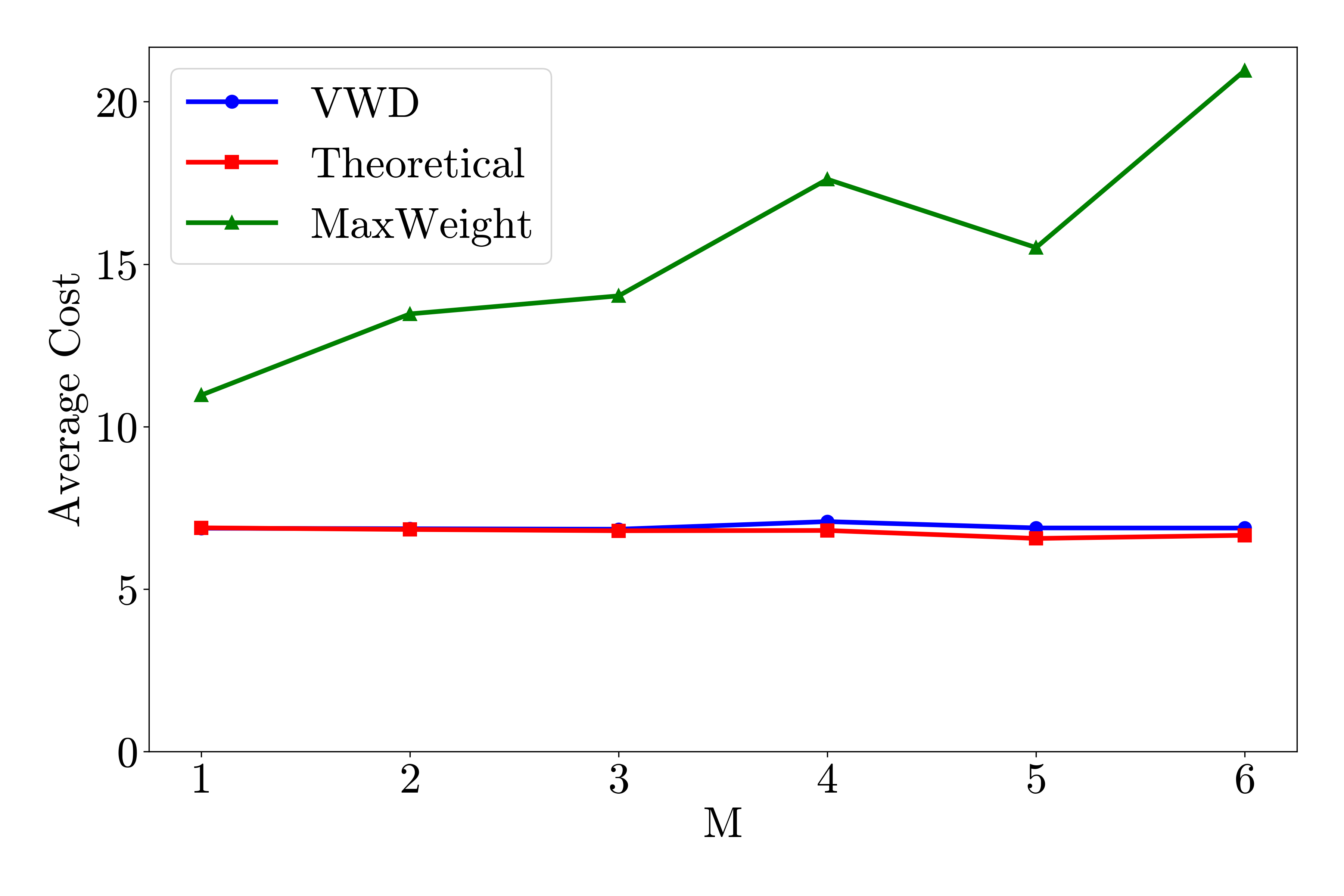}
       \label{f2:M10}
    }
    \caption{Average Cost when N/M changes.}
    \label{fig:f2_M}
\end{figure*}
\subsection{Proportional fairness in both throughput and AoI}
This is the problem described in Example~\ref{example3}, where the objective is to achieve proportional fairness in terms of both throughput and AoI. In particular, the utility function of device $i$ is defined as $\log m_i - \log h_i$, and the BS aims to maximize the total utility in the system. In addition to evaluating the empirical VWD performance and the theoretical value, we also evaluate the Random policy, which schedules devices uniformly at random in each slot.  It is well-known that the Random policy achieves proportional fairness in terms of throughput, that is, it maximizes $\sum_i \log m_i$.

We consider the case when $p_i=\frac{i}{N}$ and evaluate three cases: $N=5, 10,$ and $20$ with varying $M$, which means we select M devices at each time slot. Average Utility is shown in Fig.~\ref{fig:f3}. It can be seen from Fig.~\ref{fig:f3} that VWD is very close to the theoretical value and generally outperforms the Random policy, except when $M = N = 5$ and Eq.~\eqref{eq:eachless1} is not satisfied. In that scenario, no scheduling is required, so all policies yield the same outcome.
The difference between VWD and the Random policy highlights how the problem of proportional fairness in AoI is different from the problem of proportional fairness in throughput.
\begin{figure*}[t]
    \centering
    \subfigure[$N = 5$ devices.]{
       \includegraphics[width=0.3\linewidth]{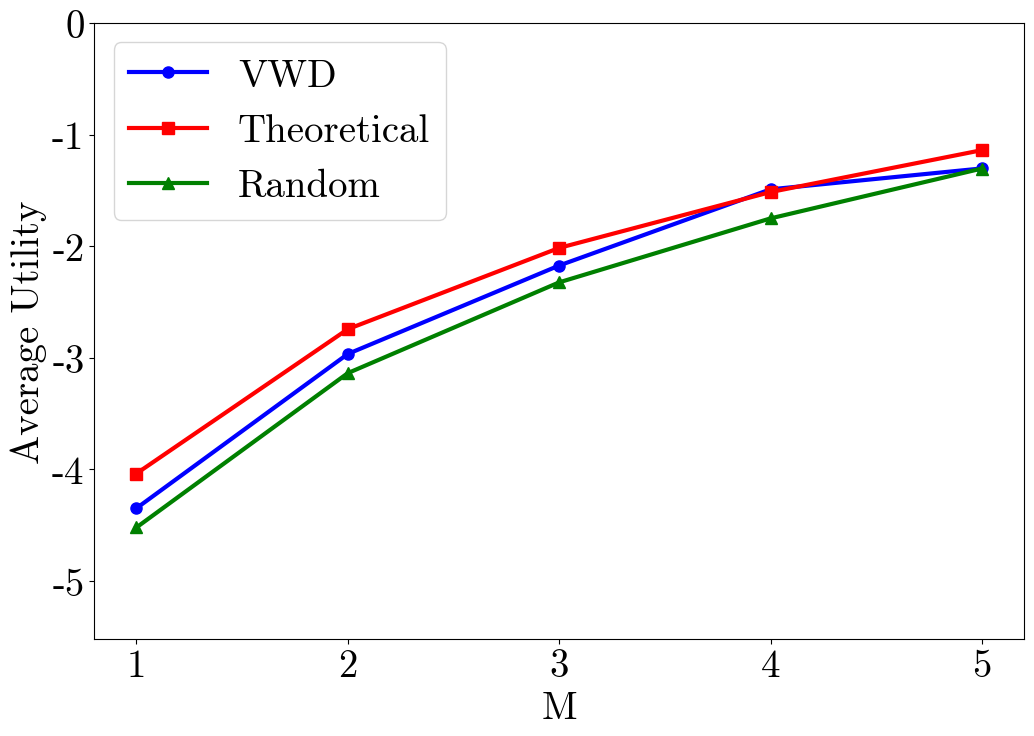}
       \label{fig:N5_3}
    }
    \hfill
    \subfigure[ $N = 10$  devices.]{
       \includegraphics[width=0.3\linewidth]{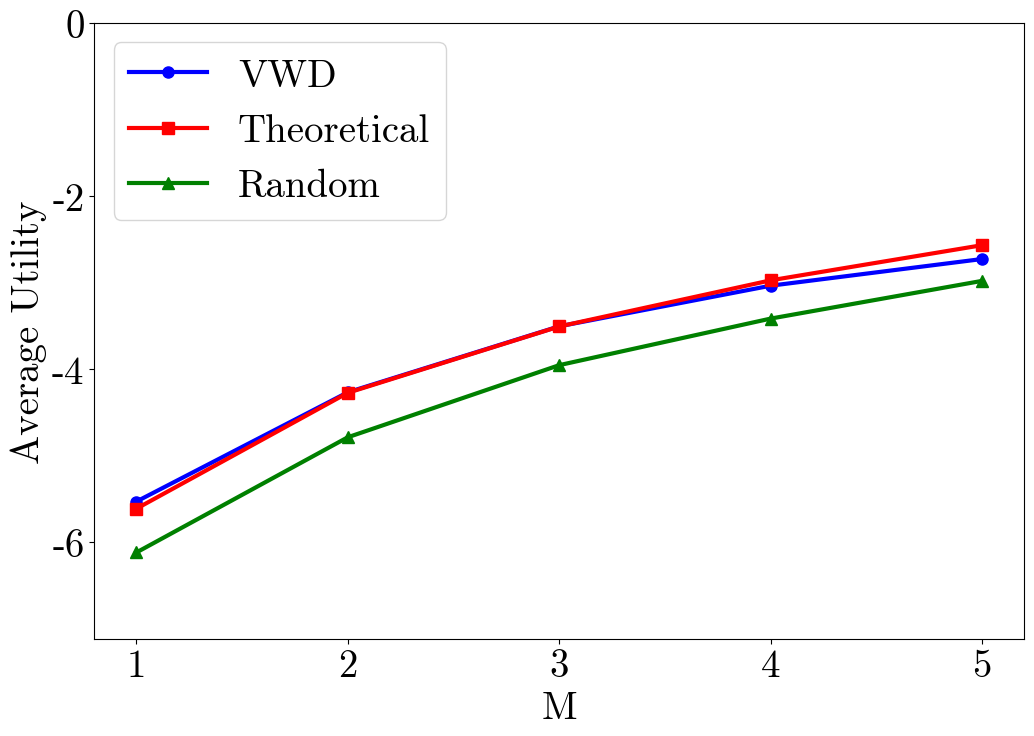}
       \label{fig:N10_3}
    }
    \hfill
    \subfigure[$N = 20$  devices.]{
       \includegraphics[width=0.3\linewidth]{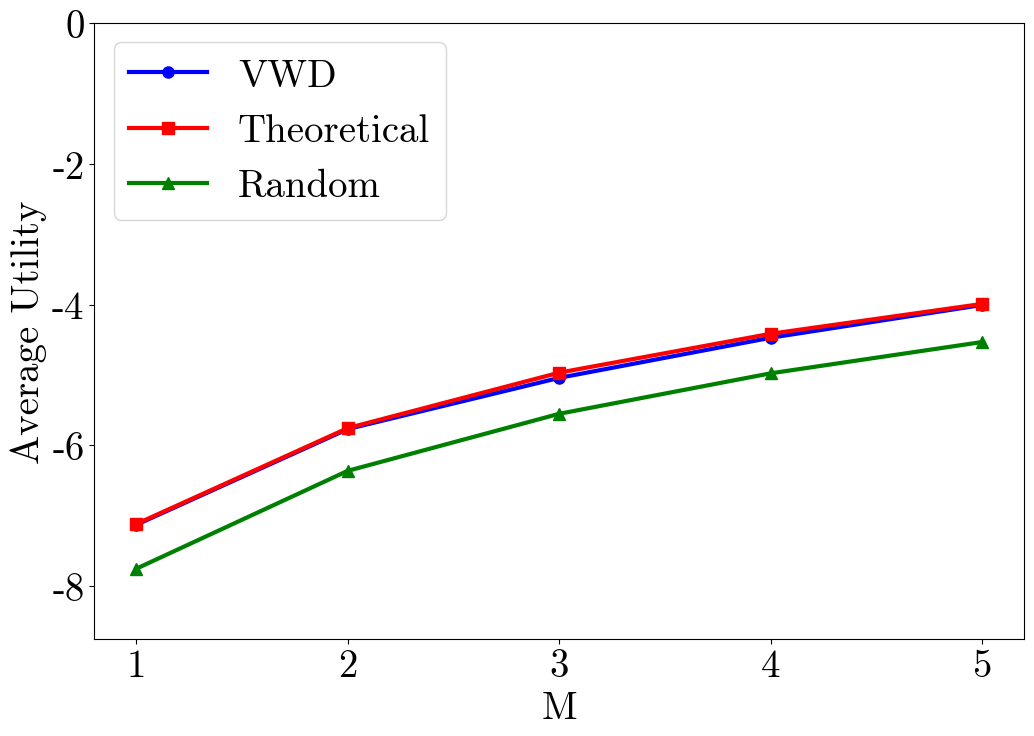}
       \label{fig:N20_3}
    }
    \caption{Average Utility when increase M}
    \label{fig:f3}
\end{figure*}
\subsection{Admission control for devices with AoI constraints}
As described in Example~\ref{example4}, the admission control problem is formulated by determining whether there exist throughput values $m_i$ such that the throughput-AoI pairs $\{(m_i, e_i) \mid 1\leq i\leq N\}$ are feasible. This formulation leverages the throughput-AoI capacity region to transform each device's strict AoI constraint into a feasibility condition on its corresponding throughput.

In the following, we consider a specific instance of the problem. Suppose that the system consists of \(N=10\) devices, where only one device is scheduled per time slot (i.e., \(M=1\)), and each device has a transmission success probability \(p_i=0.8\). For the first five devices, the AoI throughput constraint is given by \(f\) (i.e., \(h_i \leq f\)), while for the remaining five devices, the constraint is \(g\) (i.e., \(h_i \leq g\)). Consequently, the admission control problem is equivalently reformulated as determining the feasibility of the vector pairs $\{(f_i, g_i) \mid 1\leq i\leq N\}.$ The values of \(f_i\) are evenly spaced between 6 and 24, and the theoretical results are compared with the empirical results obtained from VWD.
\begin{figure}[h]
    \vspace{-10pt}
    \begin{center}
        \includegraphics[width=0.95\linewidth]{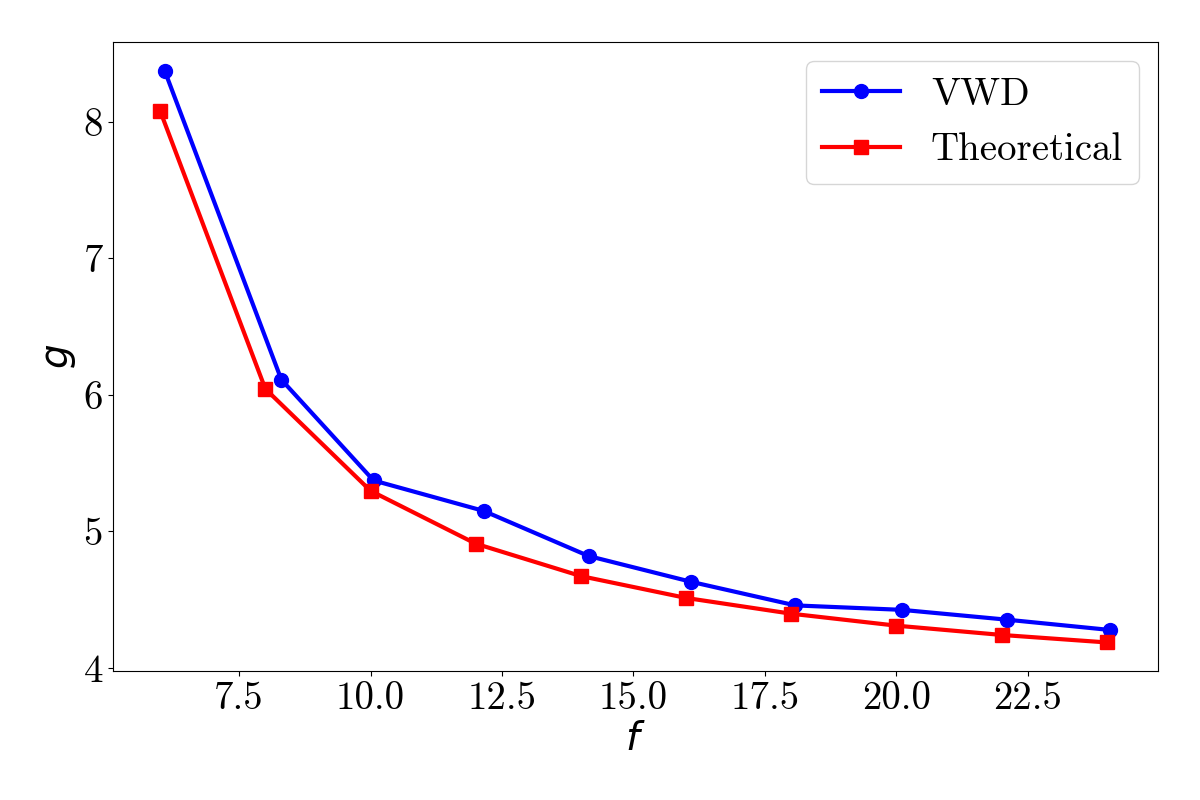}
    \end{center}
    \caption{Hard AoI constraints Capacity Region }
    \label{fig:f4}
\end{figure}

Fig.~\ref{fig:f4} clearly demonstrates that our proposed framework effectively defines the feasible vector pairs $\{(f_i, g_i) \mid 1 \leq i \leq N\}.$
Furthermore, the empirical results closely match the theoretical ones, thereby validating the efficacy of our framework in solving the admission control problem for devices subject to AoI constraints.

\section{Conclusions}\label{section:conclusion}
In this paper, we introduced a theoretical framework that characterizes the fundamental trade-off between throughput and AoI in unreliable wireless networks. Leveraging second-order approximations, we derived a concise expression for the throughput–AoI capacity region, established tight outer and inner bounds, and proposed the Variance-Weighted Deficit scheduling policy, which achieves every interior point of the capacity region. Extensive simulations in multiple application scenarios validate the framework, showing that our method often outperforms conventional approaches. These findings provide insights into the interplay between throughput and information freshness, and offer practical guidelines for designing efficient scheduling policies. Future work may extend the framework to more complex network topologies, additional performance metrics, and dynamic channel conditions.

\section{Acknowledgement}
This material is based upon work supported in part by NSF under Award Number CCF-2332800, the U.S. Army Contracting Command-Aberdeen Proving Ground under Grant Number W911NF-22-1-015, U.S. ACCDC under Grant Number W911NF2520046, and the U.S. ONR under Grant Number N000142412615. Portions of this research were conducted with the advanced computing resources provided by Texas A\&M High Performance Research Computing.

\bibliographystyle{IEEEtran}
\bibliography{reference}

\begin{thebibliography}{10}
\providecommand{\url}[1]{#1}
\csname url@samestyle\endcsname
\providecommand{\newblock}{\relax}
\providecommand{\bibinfo}[2]{#2}
\providecommand{\BIBentrySTDinterwordspacing}{\spaceskip=0pt\relax}
\providecommand{\BIBentryALTinterwordstretchfactor}{4}
\providecommand{\BIBentryALTinterwordspacing}{\spaceskip=\fontdimen2\font plus
\BIBentryALTinterwordstretchfactor\fontdimen3\font minus \fontdimen4\font\relax}
\providecommand{\BIBforeignlanguage}[2]{{%
\expandafter\ifx\csname l@#1\endcsname\relax
\typeout{** WARNING: IEEEtran.bst: No hyphenation pattern has been}%
\typeout{** loaded for the language `#1'. Using the pattern for}%
\typeout{** the default language instead.}%
\else
\language=\csname l@#1\endcsname
\fi
#2}}
\providecommand{\BIBdecl}{\relax}
\BIBdecl

\bibitem{firstaoi}
S.~Kaul, R.~Yates, and M.~Gruteser, ``Real-time status: How often should one update?'' pp. 2731--2735, 2012.

\bibitem{aoisurvey}
R.~D. Yates, Y.~Sun, D.~R. Brown, S.~K. Kaul, E.~Modiano, and S.~Ulukus, ``Age of information: An introduction and survey,'' \emph{IEEE Journal on Selected Areas in Communications}, vol.~39, no.~5, pp. 1183--1210, 2021.

\bibitem{iotaoi}
J.~Tong, L.~Fu, and Z.~Han, ``Age-of-information oriented scheduling for multichannel iot systems with correlated sources,'' \emph{IEEE Transactions on Wireless Communications}, vol.~21, no.~11, pp. 9775--9790, 2022.

\bibitem{underwateraoi}
Y.~Tian, L.~Wang, C.~Lin, Y.~Chi, B.~Lu, and Z.~Qin, ``Minimizing age of information for underwater optical wireless sensor networks,'' in \emph{IEEE INFOCOM 2023 - IEEE Conference on Computer Communications}, 2023, pp. 1--10.

\bibitem{energysys1}
Q.~Lin, J.~Su, and M.~Chen, ``Competitive online age-of-information optimization for energy harvesting systems,'' in \emph{IEEE INFOCOM 2024 - IEEE Conference on Computer Communications}, 2024, pp. 901--910.

\bibitem{energysys2}
C.~Zhao, S.~Xu, and J.~Ren, ``Aoi-aware wireless resource allocation of energy-harvesting-powered mec systems,'' \emph{IEEE Internet of Things Journal}, vol.~10, no.~9, pp. 7835--7849, 2023.

\bibitem{cpsaoi}
D.~Sinha and R.~Roy, ``Scheduling status update for optimizing age of information in the context of industrial cyber-physical system,'' \emph{IEEE Access}, vol.~7, pp. 95\,677--95\,695, 2019.

\bibitem{cpsaoi2}
F.~Yuan, S.~Tang, and D.~Liu, ``Aoi-based transmission scheduling for cyber physical systems over fading channel against eavesdropping,'' \emph{IEEE Internet of Things Journal}, vol.~11, no.~3, pp. 5455--5466, 2024.

\bibitem{uavaoi}
H.~Hu, K.~Xiong, G.~Qu, Q.~Ni, P.~Fan, and K.~B. Letaief, ``Aoi-minimal trajectory planning and data collection in uav-assisted wireless powered iot networks,'' \emph{IEEE Internet of Things Journal}, vol.~8, no.~2, pp. 1211--1223, 2021.

\bibitem{uavaoi2}
V.~Tripathi, I.~Kadota, E.~Tal, M.~S. Rahman, A.~Warren, S.~Karaman, and E.~Modiano, ``Wiswarm: Age-of-information-based wireless networking for collaborative teams of uavs,'' in \emph{IEEE INFOCOM 2023 - IEEE Conference on Computer Communications}, 2023, pp. 1--10.

\bibitem{uavaoi3}
X.~Gao, X.~Zhu, and L.~Zhai, ``Aoi-sensitive data collection in multi-uav-assisted wireless sensor networks,'' \emph{IEEE Transactions on Wireless Communications}, vol.~22, no.~8, pp. 5185--5197, 2023.

\bibitem{ramakanth2024monitoring}
R.~V. Ramakanth, V.~Tripathi, and E.~Modiano, ``Monitoring correlated sources: Aoi-based scheduling is nearly optimal,'' \emph{IEEE Transactions on Mobile Computing}, 2024.

\bibitem{ornee2019sampling}
T.~Z. Ornee and Y.~Sun, ``Sampling for remote estimation through queues: Age of information and beyond,'' in \emph{2019 International Symposium on Modeling and Optimization in Mobile, Ad Hoc, and Wireless Networks (WiOPT)}.\hskip 1em plus 0.5em minus 0.4em\relax IEEE, 2019, pp. 1--8.

\bibitem{guo2021scheduling}
D.~Guo and I.-H. Hou, ``Scheduling real-time information-update flows for the optimal confidence in estimation,'' \emph{IEEE Journal on Selected Areas in Communications}, vol.~39, no.~5, pp. 1339--1351, 2021.

\bibitem{freshcsma}
V.~Tripathi, N.~Jones, and E.~Modiano, ``Fresh-csma: A distributed protocol for minimizing age of information,'' \emph{Journal of Communications and Networks}, vol.~25, no.~5, pp. 556--569, 2023.

\bibitem{csma2}
S.~Wang, O.~T. Ajayi, and Y.~Cheng, ``An analytical approach for minimizing the age of information in a practical csma network,'' in \emph{IEEE INFOCOM 2024-IEEE Conference on Computer Communications}.\hskip 1em plus 0.5em minus 0.4em\relax IEEE, 2024, pp. 1721--1730.

\bibitem{whittleminaoi}
S.~Zhou and X.~Lin, ``An easier-to-verify sufficient condition for whittle indexability and application to aoi minimization,'' in \emph{IEEE INFOCOM 2024-IEEE Conference on Computer Communications}.\hskip 1em plus 0.5em minus 0.4em\relax IEEE, 2024, pp. 1741--1750.

\bibitem{aoithreshold}
C.~Li, S.~Li, Y.~Chen, Y.~Thomas~Hou, and W.~Lou, ``Aoi scheduling with maximum thresholds,'' in \emph{IEEE INFOCOM 2020 - IEEE Conference on Computer Communications}, 2020, pp. 436--445.

\bibitem{broadcast}
I.~Kadota, E.~Uysal-Biyikoglu, R.~Singh, and E.~Modiano, ``Minimizing the age of information in broadcast wireless networks,'' in \emph{2016 54th Annual Allerton Conference on Communication, Control, and Computing (Allerton)}, 2016, pp. 844--851.

\bibitem{stocasticarrival}
\BIBentryALTinterwordspacing
I.~Kadota and E.~Modiano, ``Minimizing the age of information in wireless networks with stochastic arrivals,'' in \emph{Proceedings of the Twentieth ACM International Symposium on Mobile Ad Hoc Networking and Computing}, ser. Mobihoc '19.\hskip 1em plus 0.5em minus 0.4em\relax New York, NY, USA: Association for Computing Machinery, 2019, p. 221–230. [Online]. Available: \url{https://doi.org/10.1145/3323679.3326520}
\BIBentrySTDinterwordspacing

\bibitem{tripathi2022optimizing}
V.~Tripathi and E.~Modiano, ``Optimizing age of information with correlated sources,'' in \emph{Proceedings of the Twenty-Third International Symposium on Theory, Algorithmic Foundations, and Protocol Design for Mobile Networks and Mobile Computing}, 2022, pp. 41--50.

\bibitem{twochannelminaoi}
\BIBentryALTinterwordspacing
J.~Pan, A.~M. Bedewy, Y.~Sun, and N.~B. Shroff, ``Minimizing age of information via scheduling over heterogeneous channels,'' in \emph{Proceedings of the Twenty-Second International Symposium on Theory, Algorithmic Foundations, and Protocol Design for Mobile Networks and Mobile Computing}, ser. MobiHoc '21.\hskip 1em plus 0.5em minus 0.4em\relax New York, NY, USA: Association for Computing Machinery, 2021, p. 111–120. [Online]. Available: \url{https://doi.org/10.1145/3466772.3467040}
\BIBentrySTDinterwordspacing

\bibitem{partialindex}
\BIBentryALTinterwordspacing
Y.~Zou, K.~T. Kim, X.~Lin, and M.~Chiang, ``Minimizing age-of-information in heterogeneous multi-channel systems: A new partial-index approach,'' in \emph{Proceedings of the Twenty-Second International Symposium on Theory, Algorithmic Foundations, and Protocol Design for Mobile Networks and Mobile Computing}, ser. MobiHoc '21.\hskip 1em plus 0.5em minus 0.4em\relax New York, NY, USA: Association for Computing Machinery, 2021, p. 11–20. [Online]. Available: \url{https://doi.org/10.1145/3466772.3467030}
\BIBentrySTDinterwordspacing

\bibitem{partialindex2}
\BIBentryALTinterwordspacing
S.~Zhou and X.~Lin, ``On the active-time condition for partial indexability and application to heterogeneous-channel aoi minimization,'' in \emph{Proceedings of the Twenty-Fifth International Symposium on Theory, Algorithmic Foundations, and Protocol Design for Mobile Networks and Mobile Computing}, ser. MobiHoc '24.\hskip 1em plus 0.5em minus 0.4em\relax New York, NY, USA: Association for Computing Machinery, 2024, p. 311–320. [Online]. Available: \url{https://doi.org/10.1145/3641512.3686371}
\BIBentrySTDinterwordspacing

\bibitem{aoiviolation}
C.~Li, Q.~Liu, S.~Li, Y.~Chen, Y.~T. Hou, and W.~Lou, ``On scheduling with aoi violation tolerance,'' in \emph{IEEE INFOCOM 2021 - IEEE Conference on Computer Communications}, 2021, pp. 1--9.

\bibitem{aoitprela}
J.~Lou, X.~Yuan, S.~Kompella, and N.-F. Tzeng, ``Aoi and throughput tradeoffs in routing-aware multi-hop wireless networks,'' in \emph{IEEE INFOCOM 2020 - IEEE Conference on Computer Communications}, 2020, pp. 476--485.

\bibitem{modiano}
I.~Kadota, A.~Sinha, and E.~Modiano, ``Optimizing age of information in wireless networks with throughput constraints,'' in \emph{IEEE INFOCOM 2018 - IEEE Conference on Computer Communications}, 2018, pp. 1844--1852.

\bibitem{aoiandtp}
E.~Fountoulakis, T.~Charalambous, A.~Ephremides, and N.~Pappas, ``Scheduling policies for aoi minimization with timely throughput constraints,'' \emph{IEEE Transactions on Communications}, vol.~71, no.~7, pp. 3905--3917, 2023.

\bibitem{daojingsecondorder}
D.~Guo, K.~Nakhleh, I.-H. Hou, S.~Kompella, and C.~Kam, ``A theory of second-order wireless network optimization and its application on aoi,'' in \emph{IEEE INFOCOM 2022 - IEEE Conference on Computer Communications}.\hskip 1em plus 0.5em minus 0.4em\relax IEEE Press, 2022, p. 999–1008.

\bibitem{guo2024aoi}
------, ``Aoi, timely-throughput, and beyond: A theory of second-order wireless network optimization,'' \emph{IEEE/ACM Transactions on Networking}, 2024.

\bibitem{siqi}
S.~Fan, Y.~Zhong, I.-H. Hou, and C.~K. Kam, ``Optimizing age of information in random access networks: A second-order approach for active/passive users,'' \emph{IEEE Transactions on Communications}, vol.~73, no.~1, pp. 439--453, 2025.

\bibitem{qostheory}
I.-H. Hou, V.~Borkar, and P.~R. Kumar, ``A theory of qos for wireless,'' in \emph{IEEE INFOCOM 2009}, 2009, pp. 486--494.

\end{thebibliography}
\end{document}